\documentclass[12pt]{amsart}
\usepackage{amssymb}
\usepackage{color}
\usepackage{amsmath,epic,curves,amscd}
\usepackage[english]{babel}
\usepackage{graphicx}
\usepackage{comment}
\usepackage{appendix}
\usepackage{mathdots}
\usepackage[all]{xy}
\usepackage{mathtools}
\usepackage{lscape}
\usepackage{stmaryrd}
\usepackage{mathabx}
\newtheorem{claim}{}[section]
\newtheorem{theorem}[claim]{Theorem}

\newtheorem{lemma}[claim]{Lemma}
\newtheorem{proposition}[claim]{Proposition}
\newtheorem{corollary}[claim]{Corollary}

\theoremstyle{remark}

\renewenvironment{proof}{\noindent{\it Proof. \hskip0pt}}
                      {$\square$\par\medskip}
\allowdisplaybreaks \textwidth 15.5 true cm \textheight 23.9 true cm
\usepackage{color}

\begin{document}
\baselineskip 6.0 truemm
\parindent 1.5 true pc

\newcommand\lan{\langle}
\newcommand\ran{\rangle}
\newcommand\tr{\operatorname{Tr}}
\newcommand\ot{\otimes}
\newcommand\ttt{{\text{\rm t}}}
\newcommand\rank{\ {\text{\rm rank of}}\ }
\newcommand\choi{{\rm C}}
\newcommand\dual{\star}
\newcommand\flip{\star}
\newcommand\cp{{{\mathbb C}{\mathbb P}}}
\newcommand\ccp{{{\mathbb C}{\mathbb C}{\mathbb P}}}
\newcommand\pos{{\mathcal P}}
\newcommand\tcone{T}
\newcommand\mcone{K}
\newcommand\superpos{{{\mathbb S\mathbb P}}}
\newcommand\blockpos{{{\mathcal B\mathcal P}}}
\newcommand\jc{{\text{\rm JC}}}
\newcommand\dec{{\mathbb D}{\mathbb E}{\mathbb C}}
\newcommand\decmat{{\mathcal D}{\mathcal E}{\mathcal C}}
\newcommand\ppt{{\mathcal P}{\mathcal P}{\mathcal T}}
\newcommand\pptmap{{\mathbb P}{\mathbb P}{\mathbb T}}
\newcommand\xxxx{\bigskip\par ================================}
\newcommand\join{\vee}
\newcommand\meet{\wedge}
\newcommand\ad{\operatorname{Ad}}
\newcommand\ldual{\varolessthan}
\newcommand\rdual{\varogreaterthan}
\newcommand{\slmp}{{\mathcal M}^{\text{\rm L}}}
\newcommand{\srmp}{{\mathcal M}^{\text{\rm R}}}
\newcommand{\smp}{{\mathcal M}}
\newcommand{\id}{{\text{\rm id}}}
\newcommand\tsum{\textstyle\sum}
\newcommand\hada{\Theta}
\newcommand\ampl{\mathbb A^{\text{\rm L}}}
\newcommand\ampr{\mathbb A^{\text{\rm R}}}
\newcommand\amp{\mathbb A}
\newcommand\rk{{\text{\rm rank}}\,}
\newcommand\calI{{\mathcal I}}
\newcommand\bfi{{\bf i}}
\newcommand\bfj{{\bf j}}
\newcommand\bfk{{\bf k}}
\newcommand\bfl{{\bf l}}
\newcommand\bfzero{{\bf 0}}
\newcommand\bfone{{\bf 1}}

\title{Compositions and tensor products of linear maps between matrix algebras}

\author{Seung-Hyeok Kye}
\address{Department of Mathematics and Institute of Mathematics, Seoul National University, Seoul 151-742, Korea}
\email{kye at snu.ac.kr}

\thanks{partially supported by NRF-2020R1A2C1A01004587, Korea}

\begin{abstract}
In this semi-expository paper, we first explain key notions from current
quantum information theory and criteria for them in a coherent way.
These include separability/entanglement, Schmidt numbers of
bi-partite states and block-positivity, together with various kinds
of positive maps between matrix algebras like entanglement breaking
maps, $k$-superpositive maps, completely positive maps, $k$-positive
maps. We will begin with concrete examples of elementary positive
maps given by $x\mapsto s^*xs$, and use Choi matrices and duality to
explain all the notions mentioned above. We also show that the Choi
matrix can be defined free from coordinates. The above notions of
positive maps give rise to mapping cones, whose dual cones are
characterized in terms of compositions or tensor products of linear maps. Through
the discussion, we exhibit an identity which connects tensor
products and compositions of linear maps between matrix algebras
through the Choi matrices. Using this identity, we show that the
description of the dual cone with tensor products is possible only
when the involving cones are mapping cones, and recover various
known criteria with ampliation for the notions mentioned above. As
another applications of the identity, we construct various mapping
cones arising from ampliation and factorization, and provide several
equivalent statements to PPT (positive partial transpose) square
conjecture in terms of tensor products.
\end{abstract}
\maketitle

\section{Introduction}

Entanglement has been considered as one of the most important notions in current quantum information theory
since its mathematical definition was given by Werner \cite{Werner-1989} for mixed states, which can be expressed
as sums of tensor products of two matrices.
Horodecki's separability criterion \cite{horo-1} tells us that positive linear maps between matrices are essential to
distinguish entanglement from separable states. Among positive maps, several classes like complete positive,
$k$-positive and $k$-superpositive maps play important roles.

The notions of complete positivity and $k$-positivity had emerged in Stinespring's representation theorem \cite{stine} in 1955, and examples
distinguishing different $k$-positivities were found by Choi \cite{choi72} and Tomiyama \cite{tom_85}.
Motivated by Woronowicz' work \cite{woronowicz} utilizing the usual duality between mapping spaces and tensor products,
the author \cite{eom-kye} found the dual objects of $k$-positivity in the tensor products of matrix algebras,
which are just Schmidt numbers in the current terminology. The notion of Schmidt numbers had been introduced by
Terhal and Horodecki \cite{terhal-sghmidt} who showed that these numbers can be determined by $k$-positive maps.
Bi-partite states of Schmidt number one are just separable states, and states which are not separable are called entangled.
The counterparts of separable states in the mapping space through Choi matrices were introduced
with various motivations \cite{{ando-04},{hsrus},{shor}} under the names; superpositive maps, entanglement breaking maps.
This notion has been extended \cite{{cw-EB},{ssz}} to $k$-superpositivity
which is just the counterpart of Schmidt numbers at most $k$.

The first purpose of this note is to explain the above notions and various known criteria for them
in a single framework. We will begin
with $k$-superpositive maps which are nonnegative sums of the maps
$x\mapsto s^*xs$ with matrices $s$ whose ranks are at most $k$.
We use the duality and Choi matrices to define $k$-positivity and Schmidt number $k$, respectively.
In this way, the original definition of $k$-positivity through ampliation
is naturally recovered, and the Choi's correspondence \cite{choi75-10} between
completely positivity of $\phi$ and the positivity of the Choi matrix $\choi_\phi$ is also obtained.
The convex cones $\superpos_k$ and $\mathbb P_k$ consisting of $k$-superpositive and $k$-positive maps, respectively,
are mapping cones as well as convex cones.
The notion of mapping cones was introduced by St\o rmer \cite{{stormer-dual}} to study the extension problem
of positive maps. It is known \cite{{sko-laa},{stormer_scand_2012}} that the dual cone of a mapping cone
can be described in terms of composition and tensor product of linear maps.

It was shown in \cite{gks} that the description of the dual cones through composition is possible
only when the involving cones are  one-sided mapping cones.
We will see that the description of dual cones through tensor products is possible
only when the involving cones are (two-sided) mapping cones.
To see this, we will exhibit the identity
\begin{equation}\label{fund}
(\phi_1\ot\phi_2)(\choi_\sigma)=\choi_{\phi_2\circ\sigma\circ\phi_1^*}
\end{equation}
which connects composition and tensor product of linear maps between matrix algebras through the Choi matrices.
This identity will be used to recover various criteria through ampliation in a single framework.
These include the characterizations of decomposable maps \cite{stormer82} and $k$-positive maps \cite{eom-kye},
together with separability criteria \cite{horo-1}, characterizations of Schmidt numbers \cite{terhal-sghmidt},
entanglement breaking maps \cite{hsrus} and $k$-superpositive maps \cite{cw-EB}.

After we explain Choi matrices and duality with a bilinear pairing, we introduce
basic notions including $k$-superpositivity, complete positivity, $k$-positivity for linear maps,
together with Schmidt numbers and $k$-blockpositivity of tensor product of matrices in the next section.
We also show that Choi matrices can be defined independent of standard matrix units to retain all the correspondences
between linear maps and tensor products. We give in Section \ref{mc_idneity} a simple proof of the identity (\ref{fund}) and
characterize mapping cones in terms of tensor products. This provides us various characterizations
of $k$-positivity in terms of tensor products. In this section, we also introduce the mapping cones of
decomposable maps and PPT maps.
Section \ref{cri_amp} will be devoted to recover various known criteria through ampliation
using the identity (\ref{fund}).

In the remainder of the paper, we exhibit two more applications of the identity (\ref{fund}).
Motivated by recent works \cite{Christandl19} and \cite{dms}, we suggest in Section \ref{constructio_MC} further
constructions of mapping cones through ampliation and factorization. Those mapping cones are dual to each other,
which may be considered as natural extensions of the duality between $k$-positivity and $k$-superpositivity.
The identity (\ref{fund}) will be useful to deal with such mapping cones. We give one more application
of the identity (\ref{fund}) in Section \ref{PPT-sec} to provide
several equivalent claims to the PPT square conjecture in terms of tensor products.
We close the paper with some questions in the final section.

This work is strongly motivated by the joint paper \cite{gks} with Erling St\o rmer and Mark Girard,
and we will follow the notations in \cite{gks}. For examples, $M_A$ denotes the algebra of all complex matrices
acting on the Hilbert space $\mathbb C^A$ with the dimension $a$.
The author is grateful to Erling St\o rmer and Mark Girard for encouragement and comments on the draft, respectively.
This is a revised version of the paper posted under the same title, which is rewritten to be self-contained
for general audiences.

\section{Convex cones in quantum information theory}

Two main tools to explain various notions from quantum information theory
are bilinear pairing to define the dual objects and the Choi matrices connecting linear maps
and tensor products of matrices. For a given linear map $\phi:M_A\to M_B$, the
{\sl Choi matrix} $\choi_\phi\in M_A\ot M_B$
is defined by
$$
\choi_\phi=\sum_{i,j} |i\ran\lan j|\ot \phi(|i\ran\lan j|)\in M_A\ot M_B.
$$
It is clear that $\phi\mapsto \choi_\phi$ is a linear isomorphism from the space $L(M_A,M_B)$
of all linear maps to $M_A\ot M_B$. We denote by $H(M_A,M_B)$ the real space of all Hermiticity preserving maps.

With the bilinear pairing $\lan a,b\ran=\tr(ab^\ttt)$ between matrices,
it is easily seen that the identity $\lan a\ot b,\choi_\phi\ran=\lan b,\phi(a)\ran$
holds for $a\in M_A$, $b\in M_B$ and $\phi\in L(M_A,M_B)$.
This is nothing but the usual bilinear pairing between dual of the tensor products and mapping spaces,
and so, it is natural to define \cite{eom-kye,woronowicz} the bilinear pairing
\begin{equation}\label{bi-map-ten}
\lan a\ot b,\phi\ran:=\lan b,\phi(a)\ran=\lan a\ot b,\choi_\phi\ran
\end{equation}
between real space $(M_A\ot M_B)^{\rm h}$ of Hermitian matrices  and the mapping space $H(M_A,M_B)$.
We also define \cite{sko-laa,gks}
the bilinear pairing between mapping spaces by
\begin{equation}\label{bi-map-map}
\lan\phi,\psi\ran:=\lan \choi_\phi,\choi_\psi\ran=\tr(\choi_\phi\choi_\psi^\ttt)
\end{equation}
for $\phi,\psi\in H(M_A,M_B)$.
Suppose that $X$ and $Y$ are finite dimensional real vector space with a bilinear pairing $\lan\ ,\ \ran$ between them.
For a subset $K\subset X$, the {\sl dual cone}  $K^\circ$ in $Y$ is defined by
the set of all $y\in Y$ satisfying $\lan x,y\ran\ge 0$ for  every $x\in K$.
It is easy to see that $K^{\circ\circ}$ is the smallest closed convex cone containing $K$.

For a linear map $\phi:M_A\to M_B$, the {\sl dual map} $\phi^*:M_B\to M_A$ is defined by
$$
\lan\phi^*(b),a\ran=\lan b,\phi(a)\ran,\qquad a\in M_A,\ b\in M_B.
$$
It is easily seen that $\choi_{\phi^*}\in M_B\ot M_A$ is the flip of $\choi_\phi\in M_A\ot M_B$.
It is also easily seen that the following identities
\begin{equation}\label{pair_dual}
\begin{aligned}
&(\phi\circ\psi)^*=\psi^*\circ\phi^*,\qquad
\lan\phi,\psi\ran=\lan\phi^*,\psi^*\ran,\\
&\lan\psi\circ\phi,\sigma\ran=\lan\phi,\psi^*\circ\sigma\ran=\lan\psi,\sigma\circ\phi^*\ran
\end{aligned}
\end{equation}
hold whenever the compositions and the bilinear pairings are defined.

We begin with the elementary positive map $\ad_s:M_A\to M_B$ defined by
$$
\ad_s (x)=s^*xs,\qquad x\in M_A,
$$
for an $a\times b$ matrix $s\in M_{A,B}$, or equivalently, a linear map $s:\mathbb C^B\to\mathbb C^A$.
It is well known \cite{stormer} that the map $\ad_s$ generates an extreme ray of the convex cone of all positive maps
from $M_A$ to $M_B$. It was also shown \cite{marcin_exp} to generate an exposed ray.
For a given $k=1,2,\dots$, we denote by
$\superpos_k$ the convex cone generated by $\ad_s$ with $\rank s\le
k$, and call a map in $\superpos_k$ {\sl $k$-superpositive} \cite{ssz}.
It is clear that the
sequence $\{\superpos_k:k=1,2,\dots\}$ is increasing, and their union coincides with
the convex cone $\superpos_{\min\{a,b\}}$.

For a matrix $s=\sum_{i,j}s_{ij}|i\ran\lan j|\in M_{A,B}$, we define the vector $\lan\tilde s|$ by
\begin{equation}\label{max}
\lan\tilde s|=\sum_{i,j}s_{ij}\lan i|\lan j|\in\mathbb C^A\ot\mathbb C^B.
\end{equation}
Then it is straightforward to check the identity
\begin{equation}\label{choi_adj}
\choi_{\ad_s}=|\tilde s\ran\lan \tilde s|
\end{equation}
holds. Indeed, we have
$$
\begin{aligned}
\choi_{\ad_s}
&=\textstyle\sum_{p,q}|p\ran\lan q|\ot \left(\textstyle\sum_{k,\ell}\bar s_{\ell k}|k\ran\lan \ell|\right)
 |p\ran\lan q| \left(\sum_{i,j} s_{ij} |i\ran\lan j|\right)\\
&=\left(\textstyle\sum_{k,\ell}\textstyle\sum_{p}\bar s_{\ell k}\lan\ell |p\ran |p\ran |k\ran\right)
 \left(\textstyle\sum_{i,j}s_{ij}\textstyle\sum_q\lan q|i\ran \lan q|\lan j|\right)\\
&=\left(\textstyle\sum_{k,\ell}\bar s_{\ell k}|\ell\ran |k\ran\right)
  \left(\textstyle\sum_{i,j} s_{ij} \lan i|\lan j|\right)
=|\tilde s\ran\lan \tilde s|.
\end{aligned}
$$
We write $\choi_K:=\{\choi_\phi\in M_A\ot M_B:\phi\in K\}$ for $K\subset L(M_A,M_B)$.
Because the rank of $s\in M_{A,B}$ coincides with the Schmidt rank of $\lan\tilde s|\in\mathbb C^A\ot\mathbb C^B$, we see that
the convex cone
$$
{\mathcal S}_k:=\choi_{\superpos_k}\subset M_A\ot M_B
$$
consists of all (unnormalized) states in the convex cone $\pos_{AB}:=(M_A\ot M_B)^+$
which are convex sums of rank one projections
onto vectors whose Schmidt ranks are at most $k$. A state in ${\mathcal S}_k\setminus {\mathcal S}_{k-1}$
is called to have {\sl Schmidt number $k$}. It is clear that ${\mathcal S}_{\min\{a,b\}}=\pos_{AB}$.

For any $|\zeta\ran=\sum_{i=1}^k |x_i\ran|y_i\ran\in\mathbb C^A\ot\mathbb C^B$ whose Schmidt rank is at most $k$, we have
$|\zeta\ran\lan\zeta|=\sum_{i,j=1}^k |x_i\ran\lan x_j|\ot |y_i\ran\lan y_j|$,
and so, we have
$$
\begin{aligned}
\lan |\zeta\ran\lan\zeta|,\phi\ran
&=\sum_{i,j=1}^k\lan |y_i\ran\lan y_j|, \phi(|x_i\ran\lan x_j|)\\
&=\left\lan \sum_{i,j=1}^k |i\ran\lan j|\ot |y_i\ran\lan y_j|, \sum_{i,j=1}^k |i\ran\lan j|\ot \phi(|x_i\ran\lan x_j|)\right\ran.
\end{aligned}
$$
Putting $|\xi\ran=\sum_{i=1}^k|i\ran|x_i\ran\in \mathbb C^k\ot\mathbb C^A$ and
$|\eta\ran=\sum_{i=1}^k|i\ran|y_i\ran\in \mathbb C^k\ot\mathbb C^B$,
we see that
\begin{equation}\label{kkkk}
\lan |\zeta\ran\lan\zeta|,\phi\ran
=\lan |\eta\ran\lan\eta|, (\id_k\ot\phi)(|\xi\ran\lan\xi|)\ran.
\end{equation}
Therefore, we see that
a map $\phi:M_A\to M_B$ belongs to ${\mathcal S}_k^\circ$ with respect to the bilinear pairing (\ref{bi-map-ten}),
or equivalently belongs to $\superpos_k^\circ$ with respect to (\ref{bi-map-map})
if and only if the map $\id_k\ot\phi$ from $M_k\ot M_A$
into $M_k\ot M_B$ is positive. We call such a map $\phi$ {\sl $k$-positive}, and denote by $\mathbb P_k$
the convex cone of all $k$-positive maps.
In short, we have seen that two mapping spaces $\superpos_k$ and $\mathbb P_k$ are dual to each other, that is, we have
$\superpos_k^\circ=\mathbb P_k$ and $\mathbb P_k^\circ=\superpos_k$.
The duality $\mathbb P_1={\mathcal S}_1^\circ$ between mapping spaces and tensor products is already implicit
in \cite{woronowicz}, and $\mathbb P_k={\mathcal S}_k^\circ$ was found in \cite{eom-kye}.
On the other hand, ${\mathcal S}_1=\mathbb P_1^\circ$ was shown in \cite{horo-1} to get a
criterion for separability.

Hermitian matrices in the convex cone
$$
\blockpos_k:=\choi_{\mathbb P_k}\subset M_A\ot M_B
$$
are called {\sl $k$-blockpositive}. Then we see that $\blockpos_k$ and $\superpos_k$ are dual to each other
with respect to (\ref{bi-map-ten}), and it follows that $\varrho\in M_A\ot M_B$ is $k$-blockpositive if and if
$\lan\xi|\varrho|\xi\ran\ge 0$ for every $|\xi\ran\in\mathbb C^A\ot\mathbb C^B$ with Schmidt rank at mots $k$.
Because $\superpos_k$ is increasing when $k$ increases,
${\mathcal S}_k$ and $\blockpos_k$ are increasing and decreasing, respectively. Because
${\mathcal S}_{\min\{a,b\}}=\pos_{AB}$ is self-dual,
we see that $\phi$ is $\min\{a,b\}$-superpositive if and only if it is $\min\{a,b\}$-positive if and only if
$\choi_\phi$ is positive. A linear map $\phi$ is called {\sl completely positive}
if $\id_k\ot\phi$ is positive for every $k=1,2,\dots$,
and we denote by $\cp_{AB}$ the convex cone of all completely positive maps from $M_A$ to $M_B$.
Because $\id_k\ot\ad_s=\ad_{I_k\ot s}$ is positive, we have the following:

\begin{theorem}\label{choi}\cite{choi75-10}
For a linear map $\phi:M_A\to M_B$, the following are equivalent:
\begin{enumerate}
\item[{\rm (i)}]
$\phi$ is completely positive,
\item[{\rm (ii)}]
$\phi$ is $\min\{a,b\}$-positive,
\item[{\rm (iii)}]
$\phi$ is $\min\{a,b\}$-superpositive,
\item[{\rm (iv)}]
$\choi_\phi$ is positive.
\end{enumerate}
\end{theorem}

Statement (iii) of Theorem \ref{choi} tells us that every completely positive map is a nonnegative sum
of $\ad_s$'s. This is called a {\sl Kraus decomposition} \cite{kraus}.
We summarize our discussion in the following diagram:

\begin{equation}\label{diagram}
\begin{matrix}
L(M_A,M_B): & \superpos_1  &\subset &\superpos_k &\subset &\cp_{AB} &\subset &\mathbb P_k &\subset &\mathbb P_1\\
\\
\phantom{\jc}\downarrow\jc  &\downarrow  &&\downarrow  &&\downarrow  &&\downarrow &&\downarrow\\
\\
M_A\ot M_B: & {\mathcal S}_1  &\subset &{\mathcal S}_k &\subset &\pos_{AB} &\subset &\blockpos_k &\subset &\blockpos_1\\
\end{matrix}
\end{equation}
Bi-partite states in ${\mathcal S}_1$ are called {\sl separable}, and a state in $\pos_{AB}$ is called
{\sl entangled} when it is not separable.

We note that Choi matrices play central roles through the discussion.
In the remainder of this section, we show that Choi matrices may be defined
independent of the choice of coordinate systems.
The Choi matrix of a linear map between matrix algebras has been considered in 1967 by de Pillis
\cite{dePillis} who showed that a linear map $\phi$ preserves Hermiticity if and only if the
Choi matrix $\choi_\phi$ is self-adjoint. The Choi matrices of positive maps and completely positive maps have been considered
by Jamio\l kowski \cite{jam_72} and Choi \cite{choi75-10}, respectively. The correspondence $\phi\mapsto\choi_\phi$ is now called
the {\sl Jamio\l kowski--Choi isomorphism}. This isomorphism has been extended for various infinite dimensional cases
\cite{{holevo_2011},{holevo_2011_a},{li_du},{Magajna_2021},{stormer_choi_mat}}, and
multi-linear maps between matrix algebras \cite{{han_kye_tri},{han_kye_optimal},{kye_multi_dual}}.
It should be noted that the correspondences in the diagram (\ref{diagram}) may be broken \cite{Paulsen_Shultz} when we
change the standard basis $\{e_{i,j}\}$ by another basis of matrix algebras
in the definition $\choi_\phi=\sum_{i,j}e_{i,j}\ot\phi(e_{i,j})$.


Now, we note that the Choi matrix $\choi_\phi\in M_A\ot M_B$ of a linear map $\phi:M_A\to M_B$ is given by
\begin{equation}\label{choi-ext}
\choi_\phi
=\sum_{i,j}e_{i,j}\ot\phi(e_{i,j})
=\sum_{i,j}(\id_A\ot\phi)(|i\ran \lan j|\ot |i\ran\lan j|)
=(\id_A\ot \phi)(|\omega\ran\lan\omega|),
\end{equation}
with $|\omega\ran=\sum_i |i\ran|i\ran\in \mathbb C^A\ot\mathbb C^A$.
In order to replace $|\omega\ran$ by another vector, we fix
a matrix $s=\sum_{i,j} s_{ij}|i\ran\lan j|\in M_A$ with the maximum rank, and replace $|\omega\ran$ by
$|\tilde s\ran\in \mathbb C^A\ot \mathbb C^A$ given by (\ref{max}).
In other words, we define
$$
\choi_\phi^s=(\id_A\ot \phi)(|\tilde s\ran\lan\tilde s|)\in M_A\ot M_B,
$$
for a linear map $\phi$ from $M_A$ to $M_B$.
Note that $|\tilde s\ran$ has the maximal Schmidt rank.
When $s$ is the identity matrix, this gives rise to the usual Choi matrix.
We note that
$\choi_\phi^s=(\id_A\ot\phi)(\choi_{\ad_s})$ by (\ref{choi_adj}), and so we get the following identity
\begin{equation}\label{choi_var}
\choi_\phi^s=(\id_A\ot \phi)(\choi_{\ad_s})
=(\id_A\ot \phi)\circ(\id_A\circ\ad_s)(|\omega\ran\lan\omega|)
=\choi_{\phi\circ\ad_s}.
\end{equation}
We also define the bilinear pairing on the real vector space $H(M_A,M_B)$ by
$$
\lan\phi,\psi\ran^s:=\lan\choi_\phi^s,\choi_\psi^s\ran,
$$
for $\phi,\psi\in H(M_A,M_B)$, following the current definition of Choi matrices.

The statement (i) of the following Theorem \ref{choi_matrix} tells us that the dual object of the mapping cone $\superpos_k$ is
independent of the choice the matrix $s$.
We also see that the notions of Schmidt number $k$ and $k$-blockpositivity,
as the Choi matrices of $k$-superpositive maps and $k$-positive maps, respectively,
do not depend on the choice of $s\in M_A$.

\begin{theorem}\label{choi_matrix}
For $k=1,2,\dots$,  $\phi\in H(M_A,M_B)$ and $s\in M_A$ with the maximal rank, we have the following:
\begin{enumerate}
\item[(i)]
$\lan\phi,\psi\ran\ge 0$ for every $\psi\in\superpos_k$ if and only if $\lan\phi,\psi\ran^s\ge 0$ for every $\psi\in\superpos_k$,
\item[(ii)]
$\{\choi_\phi:\phi\in\superpos_k\}=\{\choi_\phi^s:\phi\in\superpos_k\}$,
\item[(iii)]
$\{\choi_\phi:\phi\in \superpos_k^\circ\}=\{\choi_\phi^s:\phi\in \superpos_k^\circ\}$.
\end{enumerate}
\end{theorem}

\begin{proof}
We first note that $\ad_s^*=\ad_{s^\ttt}$. In fact, we have
$$
\lan \ad_s^*(y),x\ran
=\lan y,\ad_s(x)\ran
=\lan y, s^*xs\ran
=\lan \bar sys^\ttt,x\ran
=\lan\ad_{s^\ttt}(y),x\ran,
$$
for $x\in M_A$ and $y\in M_B$. It follows that
$$
\lan\phi,\psi\ran^s=\lan\phi\circ\ad_s,\psi\circ\ad_s\ran =\lan \phi, \psi\circ\ad_s\circ\ad_{s^\ttt}\ran=\lan\phi,\psi\circ\ad_{s^\ttt s}\ran,
$$
by (\ref{choi_var}). Therefore, we see that
$\lan\phi,\psi\ran^s\ge 0$ holds for every $\psi\in\superpos_k$ if and only if
$\lan\phi,\psi\circ\ad_{s^\ttt s}\ran \ge 0$ for every $\psi\in\superpos_k$ if and only if
$\lan\phi,\sigma\ran \ge 0$ for every $\sigma\in\superpos_k \circ\ad_{s^\ttt s}$.
Because $s^\ttt s$ is nonsingular, we have $\superpos_k=\superpos_k\circ\{\ad_{s^\ttt s}\}$ for $k=1,2,\dots$,
and this proves (i).
We also have
$$
\{\choi_\phi^s:\phi\in\superpos_k\}=\{\choi_{\phi\circ\ad_s}:\phi\in\superpos_k\}
=\{\choi_\sigma:\sigma\in\superpos_k\circ\ad_s\},
$$
by (\ref{choi_var}). This proves (ii) since $\superpos_k=\superpos_k\circ\ad_s$ by non-singularity of $s$.
For (iii), it suffices to show $\superpos_k^\circ=\superpos_k^\circ\circ\ad_s$, or equivalently,
$\phi\in\superpos_k^\circ$ if and only if $\phi\circ\ad_{s^{-1}}\in\superpos_k^\circ$.
Indeed, we have
$$
\begin{aligned}
\phi\circ\ad_{s^{-1}}\in\superpos_k^\circ\
&\Longleftrightarrow\ \lan\phi\circ\ad_{s^{-1}},\psi\ran\ge 0\ {\text{\rm for every}}\ \psi\in\superpos_k\\
&\Longleftrightarrow\ \lan\phi\circ\ad_{s^{-1}},\psi\circ\ad_{s^{-1}}\ran\ge 0\ {\text{\rm for every}}\ \psi\in\superpos_k\\
&\Longleftrightarrow\ \lan\phi, \psi\ran^{s^{-1}}\ge 0\ {\text{\rm for every}}\ \psi\in\superpos_k\\
&\Longleftrightarrow\ \phi\in\superpos_k^\circ,
\end{aligned}
$$
by (i). This completes the proof.
\end{proof}

After the author had posted the first version of this paper, he considered more general situations
to replace $|\omega\ran\lan\omega|$ in (\ref{choi-ext})
by arbitrary $\Sigma\in M_A\ot M_A$, and defined
$$
\choi^\Sigma_\phi=(\id_A\ot \phi)(\Sigma)\in M_A\ot M_B.
$$
It was shown in \cite{kye-choi_mat} that the Choi's correspondence between completely positivity of $\phi$ and positivity
of $\choi^\Sigma_\phi$ is retained if and only if
$\Sigma$ is a positive rank one matrix whose range vector has the full Schmidt rank,
that is, $\Sigma=|\tilde s\ran\lan\tilde s|$ for a matrix $s$ of full rank.

\section{Mapping cones and tensor products of linear maps}\label{mc_idneity}

A closed convex cone $K$ of positive linear maps in $H(M_A,M_B)$ is called a  {\sl right mapping cone} \cite{gks}
if $K\circ\cp_{AA}\subset K$ holds,
where $K_1\circ K_2$ denotes the set of all $\phi_1\circ\phi_2$ with $\phi_i\in K_i$ for $i=1,2$.
In fact, this is the case if and only if $K\circ\cp_{AA}= K$ holds, since $\id_A\in\cp_{AA}$.
{\sl Left mapping cones} are defined similarly. A closed convex cone
$K$ is also called a {\sl mapping cone} when $\cp_{BB}\circ K\circ\cp_{AA}\subset K$ holds. A convex cone is a mapping cone
if and only if it is both left and right mapping cone.


Suppose that $K$ is a convex cone with $\superpos_1\subset K\subset \mathbb P_1$. Then we see that
$K$ is a right mapping cone if and only if $\lan\phi\circ\sigma,\psi\ran\ge 0$ holds for every $\phi\in K$,
$\sigma\in\cp_{AA}$ and $\psi\in K^\circ$
if and only if $\lan\phi,\psi\circ\sigma^*\ran\ge 0$ holds for every $\phi\in K$, $\sigma\in\cp_{AA}$ and $\psi\in K^\circ$
if and only if $K^\circ$ is a right mapping cone, since $\cp_{AA}$ is self-dual.
Similarly, we also see that $K$ is a right mapping cone if and only if $K^*=\{\phi^*:\phi\in K\}$ is a left mapping cone.
It is clear that $\superpos_k$ is a mapping cone, and so $\mathbb P_k$ is also a mapping cone.

When $K$ is a mapping cone in $H(M_A,M_A)$ with minor additional assumptions, it was shown in \cite{{sko-laa},{stormer_scand_2012}}
that the following statements for a linear
map $\phi$ are equivalent:
\begin{enumerate}
\item[(MC1)]
$\phi\in K^\circ$,
\item[(MC2)]
$\phi\circ\psi$ is completely positive for every $\psi\in K$,
\item[(MC3)]
$\psi\ot\phi$ is positive for every $\psi\in K$,
\item[(MC4)]
$(\psi\ot\phi)(\choi_\id)\ge 0$ for every for every $\psi\in K$.
\end{enumerate}

We may use the identity (\ref{pair_dual}) to see that the equivalence (MC1)
$\Longleftrightarrow$ (MC2) implies that $K$ is a left mapping cone.
In fact, we see that $\phi\in (\cp_{BB}\circ K)^\circ$
if and only if $\lan \phi,\sigma\circ\psi\ran\ge 0$ for every $\psi\in K$ and $\sigma\in\cp_{BB}$
if and only if $\lan \phi\circ\psi^*,\sigma\ran\ge 0$ for every $\psi\in K$ and $\sigma\in\cp_{BB}$
if and only if $\phi\circ\psi^*\in\cp_{BB}$  for every $\psi\in K$, and so we have
$$
(\cp_{BB}\circ K)^\circ =\{\phi\in H(M_A,M_B): \phi\circ\psi^*\in\cp_{BB}\ {\text{\rm for every}}\ \psi\in K\},
$$
for a convex cone $K\subset H(M_A,M_B)$.
The dual cone $(K\circ\cp_{AA})^\circ$ can be handled in the same way, and we have the following:

\begin{theorem}\cite{gks}
Suppose that $K$ is a closed convex cone of positive maps in $H(M_A,M_B)$. Then we have the following:
\begin{enumerate}
\item[(i)]
$K$ is a left mapping cone if and only if
$$
K^\circ =\{\phi\in H(M_A,M_B): \phi\circ\psi^*\in\cp_{BB}\ {\text{\rm for every}}\ \psi\in K\},
$$
\item[(ii)]
$K$ is a right mapping cone if and only if
$$
K^\circ =\{\phi\in H(M_A,M_B): \psi^*\circ\phi \in\cp_{AA}\ {\text{\rm for every}}\ \psi\in K\}.
$$
\end{enumerate}
\end{theorem}

In order to investigate the relations between [MC2], [MC3] and [MC4], we need the identity (\ref{fund}) which relates
compositions and tensor products of linear maps.
We suppose that $\phi_i:M_{A_i}\to M_{B_i}$ is a linear map
between matrix algebras $M_{A_i}$ and $M_{B_i}$, for $i=1,2$.
Note that every element in the tensor product $M_{A_1}\ot M_{A_2}$ is expressed as the Choi matrix $\choi_\sigma$ of a linear map
$\sigma:M_{A_1}\to M_{A_2}$, and the composition
$\phi_2\circ\sigma\circ\phi_1^*$ maps $M_{B_1}$ into $M_{B_2}$. In this circumstance, we will show that the identity
(\ref{fund}) holds.  See {\sc Figure 1}.

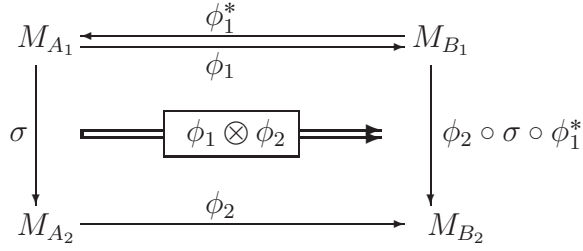
\begin{figure}\label{fig1}
\begin{center}
\setlength{\unitlength}{0.05 truecm}
\begin{picture}(140,75)
\put (0,0){$M_{A_2}$}
\put (17,3){\vector (1,0){86}}
\put (109,0){$M_{B_2}$}
\put (0,50){$M_{A_1}$}
\put (17,50){\vector (1,0){86}}
\put (103,53){\vector (-1,0){86}}
\put (105,50){$M_{B_1}$}
\put (5,45){\vector(0,-1){37}}
\put (-2,25){$\sigma$}
\put (50,56){$\phi_1^*$}
\put (50,43){$\phi_1$}
\put (50,6){$\phi_2$}
\put (110,45){\vector(0,-1){37}}
\put (113,25){$\phi_2\circ\sigma\circ\phi_1^*$}
\put (39,21){\line(1,0){36}}
\put (75,21){\line(0,1){12}}
\put (75,33){\line(-1,0){36}}
\put (39,33){\line(0,-1){12}}
\thicklines
\put (45, 25){$\phi_1\ot\phi_2$}
\put (17,26){\line(1,0){22}}
\put (17.3,25.9){\line(0,1){2}}
\put (75, 26){\vector(1,0){22}}
\put (17,28){\line(1,0){22}}
\put (75, 28){\vector(1,0){22}}
\end{picture}
\end{center}
\caption{The map $\phi_1\ot \phi_2$ sends $\choi_\sigma$ to $\choi_{\phi_2\circ\sigma\circ\phi_1^*}$.}
\end{figure}

To prove the identity (\ref{fund}), we take $b_i\in M_{B_i}$ with $i=1,2$. Then
we have
$$
\begin{aligned}
\lan b_1\ot b_2, \choi_{\phi_2\circ\sigma\circ\phi_1^*}\ran_{B_1B_2}
&=\lan b_2, \phi_2(\sigma(\phi_1^*(b_1)))\ran_{B_2}\\
&=\lan \sigma^*(\phi_2^*(b_2)),\phi_1^*(b_1)\ran_{A_1}\\
&=\tsum_{i,j}\lan\sigma^*(\phi_2^*(b_2)),e_{i,j}^{A_1}\ran_{A_1}\lan\phi_1^*(b_1),e_{i,j}^{A_1}\ran_{A_1}\\
&=\tsum_{i,j}\lan b_2, \phi_2(\sigma(e_{i,j}^{A_1}))\ran_{B_2}\lan b_1, \phi_1(e_{i,j}^{A_1})\ran_{B_1}\\
&=\tsum_{i,j}\lan b_1\ot b_2, \phi_1(e_{i,j}^{A_1})\ot
\phi_2(\sigma(e_{i,j}^{A_1}))\ran_{B_1B_2},
\end{aligned}
$$
where $\{e^A_{i,j}\}$ denotes the matrix units of $M_A$.
Therefore, it follows that
$$
\begin{aligned}
\choi_{\phi_2\circ\sigma\circ\phi_1^*}
&=\tsum_{i,j}\phi_1(e_{i,j}^A)\ot \phi_2(\sigma(e_{i,j}^A))\\
&=\tsum_{i,j}(\phi_1\ot\phi_2)(e_{i,j}^A\ot \sigma(e_{i,j}^A))\\
&=(\phi_1\ot\phi_2)(\choi_\sigma),
\end{aligned}
$$
as it was required.
If we plug $\phi_1=\id_A$, $\phi_2=\phi:M_A\to M_B$ and
$\sigma=\ad_s$ in (\ref{fund}), then we recover the identity (\ref{choi_var}).
In case when $A_1=A_2$, $B_1=B_2$ and
$\sigma:M_A\to M_A$ is the identity map, we have
$$
\choi_{\phi_2\circ\phi_1^*}=(\phi_1\ot\phi_2)(\choi_\id)
$$
which recovers the relation (12) of \cite{gks} and suggests relations between [MC2] and [MC4].

In order to investigate the conditions [MC3] and [MC4], we consider the following diagrams
$$
\begin{CD}
M_A @>\phi>> M_B\\
@VV \tau V \\
M_A @>\psi>> M_B
\end{CD}
\qquad\qquad\qquad
\begin{CD}
M_A @>\phi>> M_B\\
@VV \id_A V \\
M_A @>\psi>> M_B
\end{CD}
\qquad\qquad\qquad
\begin{CD}
M_B @>\phi^*>> M_A\\
@VV \id_B V \\
M_B @>\psi^*>> M_A
\end{CD}
$$
in {\sc Figure 1}. Applying the identity (\ref{fund}), we have the following identities
$$
\begin{aligned}
\lan\sigma^*\circ\psi\circ\tau,\phi\ran
&=\lan\psi\circ\tau,\sigma\circ\phi\ran
=\lan\psi\circ\tau\circ\phi^*,\sigma\ran
=\lan(\phi\ot\psi)(\choi_\tau),\choi_\sigma\ran,\\
\lan\phi,\sigma^*\circ\psi\ran
&=\lan\phi^*, \psi^*\circ\sigma\ran
=\lan\psi\circ\phi^*,\sigma\ran =\lan(\phi\ot\psi)(\choi_{\id_A}),\choi_{\sigma}\ran,\\
\lan\phi,\psi\circ\sigma\ran
&=\lan\psi^*\circ\phi,\sigma\ran =\lan(\phi^*\ot\psi^*)(\choi_{\id_B}),\choi_\sigma\ran.
\end{aligned}
$$
The first line tells us  that $\phi\in (\cp_{BB}\circ K\circ\cp_{AA})^\circ$ if and only if $\phi\ot\psi$ sends positive matrices
to positive matrices for every $\psi\in K$, that is, $\phi\ot\psi$ is a positive map for every $\psi\in K$.
In other words, we have
\begin{equation}\label{ghjhvykyu}
(\cp_{BB}\circ K\circ\cp_{AA})^\circ
=\{\phi\in H(M_A,M_B): \phi\ot\psi\ {\text{\rm is positive for every}}\ \psi\in K\},
\end{equation}
for a convex cone $K\subset H(M_A,M_B)$.
The exactly same argument with the second and third identities may be applied to dual cones
$(\cp_{BB}\circ K)^\circ$ and $(K\circ\cp_{AA})^\circ$,
to get the following:
$$
\begin{aligned}
(\cp_{BB}\circ K)^\circ
&=\{\phi\in H(M_A,M_B): (\phi\ot\psi)(\choi_{\id_A})\in\pos_{BB}\ {\text{\rm for every}}\ \psi\in K\},\\
(K\circ\cp_{AA})^\circ
&=\{\phi\in H(M_A,M_B): (\phi^*\ot\psi^*)(\choi_{\id_B})\in\pos_{AA}\ {\text{\rm  for every}}\ \psi\in K\}.
\end{aligned}
$$
Therefore, we have the following:

\begin{theorem}
For a closed convex cone $K$ of positive maps from $M_A$ into $M_B$, we have the following:
\begin{enumerate}
\item[(i)]
$K$ is a mapping cone if and only if
$$
K^\circ =\{\phi\in H(M_A,M_B): \phi\ot\psi\ {\text{\rm is positive for every}}\ \psi\in K\},
$$
\item[(ii)]
$K$ is a left mapping cone if and only if
$$
K^\circ=\{\phi\in H(M_A,M_B): (\phi\ot\psi)(\choi_{\id_A})\in\pos_{BB}\ {\text{\rm for every}}\ \psi\in K\},
$$
\item[(iii)]
$K$ is a right mapping cone if and only if
$$
K^\circ=\{\phi\in H(M_A,M_B): (\phi^*\ot\psi^*)(\choi_{\id_B})\in\pos_{AA}\ {\text{\rm  for every}}\ \psi\in K\}.
$$
\end{enumerate}
\end{theorem}

Therefore, we see that $\phi$ is $k$-positive, that is, $\id_k\ot
\phi$ is positive if and only if $\psi\ot\phi$ is positive for every
$\psi\in\superpos_k$. Since $\id_k$ is a typical example of
$\superpos_k$, one may suspect if $\id_k$ in the definition of
$k$-positivity may be replaced by another $k$-superpositive map.
When we fix a matrix $s$ with rank $k$, it is easy to see that
$\ad_s\ot\phi$ is positive if and only if $\id_k\ot\phi$ is positive
using singular value decomposition of $s$.
This is also due to the fact that
$\superpos_k$ is singly generated as a mapping cone.

\begin{proposition}\label{single}
Let $s\in M_{A,B}$ be of rank $k$. Then we have
$(\cp\circ\{\ad_s\}\circ\cp)^{\circ\circ}=\superpos_k$.
\end{proposition}

\begin{proof}
We first note $\cp\circ\{\ad_s\}\circ\cp\subset \superpos_k$, which
implies
$(\cp\circ\{\ad_s\}\circ\cp)^{\circ\circ}\subseteq\superpos_k$. For
the reverse inclusion, it suffices to show that $\ad_a\in
\cp\circ\{\ad_s\}\circ\cp$ whenever $\rank a\le k$, because every
map in $\superpos_k$ is the sum of such maps. Write $\rank a=\ell\le
k$. By singular value decomposition, we can take $v_1:\mathbb
C^k\to\mathbb C^B$ and $v_2:\mathbb C^k\to\mathbb C^A$ such that
\begin{equation}\label{svd}
s=v_2d_1v_1^*,\quad v_1^*v_1=v_2^*v_2=\id_{\mathbb C^k},
\end{equation}
where $d_1$ is a $k\times k$ diagonal matrix with positive real
diagonal entries. We also take $v_3:\mathbb C^\ell\to\mathbb C^B$
and $v_4:\mathbb C^\ell\to\mathbb C^A$ such that
$$
a=v_4d_2v_3^*,\quad v_3^*v_3=v_4^*v_4=\id_{\mathbb C^\ell},
$$
where $d_2$ is an $\ell\times\ell$ diagonal matrix with positive
real diagonal entries. Because $\ell\le k$, we may take $w:\mathbb
C^\ell\to\mathbb C^k$ so that $d_2=w^*d_1w$. Then we have
$$
a=v_4d_2v_3^*=v_4w^*d_1wv_3^*=(v_4w^*v_2^*)s(v_1wv_3^*).
$$
Therefore, we see that $\ad_a =
\ad_{v_1wv_3^*}\circ\ad_s\circ\ad_{v_4w^*v_2^*}$ belongs to
$\cp\circ\{\ad_s\}\circ\cp$.
\end{proof}

Now, we fix $s\in M_{A,B}$ of rank $k$.
By Proposition \ref{single}, we see that
$\phi\in\superpos_k^\circ$ if and only if
$\phi\in(\cp\circ\{\ad_s\}\circ\cp)^\circ$. We apply the first identity of
(\ref{ghjhvykyu}) to the convex cone generated by $\{\ad_s\}$, to
conclude that $\phi\in(\cp\circ\{\ad_s\}\circ\cp)^\circ$ holds if and only if $\phi\ot \ad_s$ is positive,
or equivalently, $\ad_s\ot\, \phi$ is positive.

One more important example of a positive map is the transpose map $\ttt$. If $K$ is a mapping cone then $\{\phi\circ\ttt:\phi\in K\}$
is also a mapping cone. Especially,
$$
\ccp:=\cp\circ\ttt
$$
is a mapping cone whose element is called {\sl completely copositive}.
If $K_1$ and $K_2$ are mapping cones then their convex hull $K_1\join K_2$ and intersection $K_1\meet K_2$ are also mapping cones.
In this way, we have mapping cones
$$
\dec:=\cp\join\ccp,\qquad
\pptmap:=\cp\meet\ccp.
$$
The {\sl partial transpose} in $M_A\ot M_B$ is defined by $(a\ot b)^\Gamma=a^\ttt\ot b$.
Because $\choi_{\phi\circ\ttt}=(\choi_\phi)^\Gamma$,
we see that $\phi\in\pptmap$ if and only if both $\choi_\phi$ and $\choi_\phi^\Gamma$ are positive. A state $\varrho$
belongs to the convex cone
$$
\ppt:=\choi_{\pptmap}=\choi_\cp\meet \choi_\ccp
$$
if and only if $\varrho^\Gamma$ is positive, and such states are called {\sl of positive partial transpose (PPT)}.
We also denote
$$
\decmat:=\choi_{\dec}.
$$
It is clear that $\dec$ and $\pptmap$ are dual to each other,
and $\dec\subset\mathbb P_1$. By duality, we also have ${\mathcal S}_1\subset\ppt$, that is,
every separable state is of PPT \cite{{choi-ppt},{peres}}.
Maps in $\dec$ and $\pptmap$ are called  {\sl decomposable maps} and {\sl PPT maps}, respectively.
Now, we summarize as follows:
\begin{equation}\label{chain_ppt_dec}
\begin{matrix}
L(M_A,M_B): & \superpos_1  &\subset &\pptmap &\subset &\cp_{AB} &\subset &\dec &\subset &\mathbb P_1\\
\\
\phantom{\jc}\downarrow\jc  &\downarrow  &&\downarrow  &&\downarrow  &&\downarrow &&\downarrow\\
\\
M_A\ot M_B: & {\mathcal S}_1  &\subset &\ppt &\subset &\pos_{AB} &\subset &\decmat &\subset &\blockpos_1\\
\end{matrix}
\end{equation}

Comparing two chains of mapping cones in the diagrams
(\ref{diagram}) and (\ref{chain_ppt_dec}), it is natural to ask if there are any inclusion relations between
$\dec$ and $\mathbb P_k$, or equivalently those between $\ppt$ and ${\mathcal S}_k$.
In case of $(a,b)=(2,2)$, it is known \cite{stormer} that
$$
\mathbb P_1[M_2,M_2]=\dec[M_2,M_2]
$$
with the obvious notations,
that is, every positive map between $M_2$ is decomposable.
See also \cite{{book_au_Sz},{stormer_book}} for another proofs.
It was also shown in \cite{woronowicz} that
$$
\mathbb P_1[M_2,M_3]=\dec[M_2,M_3],\qquad
\mathbb P_1[M_3,M_2]=\dec[M_3,M_2]
$$
together with $\mathbb P_1[M_2,M_4]\supsetneqq \dec[M_2,M_4]$.
The first example of an indecomposable positive map
was found in \cite{choi75-10} when $(a,b)=(3,3)$. In this case,  it had been known
that $\mathbb P_2\subset\dec$ holds
for a special class of positive maps \cite{cho-kye-lee},
and the dual claim $\ppt\subset{\mathcal S}_2$ in $M_3\ot M_3$ was also conjectured in \cite{sbl}.
See a survey article \cite{kye_rev}.
The relation
$$
\mathbb P_2[M_3,M_3]\subset\dec[M_3,M_3]
$$
was shown in \cite{ylt} to be true.
It is now known \cite{bhat-osaka-2020} that a $2$-positive map need not to be decomposable in general,
or equivalently, the Schmidt number of a PPT state may exceed two. In fact,
it is known \cite{{criello},{chen_yany_tang-2017},{hllm}}
that the Schmidt number of a PPT state may be arbitrary large. The relation
\begin{equation}\label{dec-con}
\mathbb P_{n-1}[M_n,M_n]\subset \dec[M_n,M_n]
\end{equation}
is conjectured in \cite{Christandl19}, in connection with the PPT square conjecture which will be explained
in Section \ref{PPT-sec}.

\section{Criteria using ampliation}\label{cri_amp}

St\o rmer \cite{stormer82} showed that a map $\phi:M_A\to M_B$ is decomposable if and only if
$\id_k\ot\phi$ sends every PPT matrix to a positive matrix for every $k=1,2,\dots$.
Motivated by this results, the author \cite{eom-kye} showed that
$\phi$ is $k$-positive if and only if $\phi\ot\id_B$ sends ${\mathcal S}_k$ to positive matrices.
Ampliation is also useful to characterize properties of states in $M_A\ot M_B$.
It was shown by Horodecki's \cite{horo-1} that a state $\varrho\in \pos_{AB}$ is separable
if and only if $(\id_A\ot \phi)(\varrho)\in\pos_{AA}$ for every positive map $\phi:M_B\to M_A$.
It is also shown by Terhal and Horodecki \cite{terhal-sghmidt} that
$\varrho\in \pos_{AA}$ has Schmidt number at most $k$ if and only if
$(\id_A\ot \phi)(\varrho)\in\pos_{AA}$ for every $k$-positive map $\phi:M_A\to M_A$.
Furthermore, it was shown in \cite{hsrus} that $\choi_\phi$ is separable if and only if
$\id_A\ot\phi$ sends every states to a separable states.
In this section, we use the identity (\ref{fund}) in a systematic way to recover the above criteria using ampliation.

We first investigate the role of the ampliation maps $\id_A\ot\phi$ and $\phi\ot \id_B$.
To do this, we consider the following diagrams
$$
\begin{CD}
M_A @>\id_A>> M_A\\
@VV \sigma^* V \\
M_A @>\phi>> M_B
\end{CD}
\qquad\qquad\qquad
\begin{CD}
M_A @>\phi>> M_B\\
@VV \psi V \\
M_B @>\id_B>> M_B
\end{CD}
$$
in {\sc Figure 1}, to get the identities
$$
\begin{aligned}
\lan\phi,\psi\circ\sigma\ran
&=\lan \phi\circ\sigma^*,\psi\ran
=\lan \choi_{\phi\circ\sigma^*},\choi_\psi\ran
=\lan(\id_A\ot\phi)(\choi_{\sigma^*}),\choi_\psi\ran,\\
\lan\phi,\sigma\circ\psi\ran
&=\lan\psi\circ\phi^*,\sigma^*\ran
=\lan\choi_{\psi\circ\phi^*},\choi_{\sigma^*}\ran
=\lan(\phi\ot\id_B)(\choi_{\psi}),\choi_{\sigma^*}\ran.
\end{aligned}
$$
From the first line, we see that $\phi\in(K\circ\cp_{AA})^\circ$ if and only if the ampliation $\id_A\ot\phi$ sends $\pos_{AA}$
to $\choi_{K^\circ}$. We also see that $\phi\in(\cp_{BB}\circ K)^\circ$ if and only if $\phi\ot\id_B$ sends
$\choi_{K}$ to $\pos_{BB}$
from the second line. In other words, we have
\begin{equation}\label{amp_buyvyjj}
\begin{aligned}
(K\circ\cp_{AA})^\circ
&= \{\phi\in H(M_A,M_B): (\id_A\ot\phi)(\varrho)\in \choi_{K^\circ}\ {\text{\rm for every}}\ \varrho\in\pos_{AA}\},\\
(\cp_{BB}\circ K)^\circ
&= \{\phi\in H(M_A,M_B): (\phi\ot\id_B)(\varrho)\in\pos_{BB}\ {\text{\rm for every}}\ \varrho\in \choi_{K}\}.
\end{aligned}
\end{equation}
Therefore, we have the following:

\begin{theorem}\label{RMP1}\cite{gks}
For a closed convex cone $K\subset H(M_A,M_B)$, the following are equivalent:
\begin{enumerate}
\item[{\rm (i)}]
$K$ is a right mapping cone,
\item[{\rm (ii)}]
$\phi\in K^\circ$ if and only if $\id_A\ot\phi$ sends positive matrices into $\choi_{K^\circ}$.
\end{enumerate}
The following are also equivalent:
\begin{enumerate}
\item[{\rm (iii)}]
$K$ is a left mapping cone,
\item[{\rm (iv)}]
$\phi\in K^\circ$ if and only if $\phi\ot\id_B$ sends $\choi_{K}$ to positive matrices.
\end{enumerate}
\end{theorem}

The statement (ii) of Theorem \ref{RMP1} with $K^\circ=\superpos_1$ shows the following result
which recovers the definition of entanglement breaking maps:

\begin{corollary}\cite{hsrus}
A linear map $\phi:M_A\to M_B$ is $1$-superpositive if and only if $\id_A\ot\phi$ send every state in $M_A\ot M_B$ to a separable state.
\end{corollary}

On the other hands, the statement (iv) with $K^\circ=\dec$ and $K^\circ=\mathbb P_k$ gives rise to the following:

\begin{corollary}\label{cor-st_ky} \cite{{stormer82},{eom-kye}}
For $\phi:M_A\to M_B$, we have the following:
\begin{enumerate}
\item[{\rm (i)}]
$\phi$ is decomposable if and only if $\phi\ot\id_B$ sends PPT states to positive matrices,
\item[{\rm (ii)}]
$\phi$ is $k$-positive if and only if $\phi\ot\id_B$ sends states with Schmidt numbers at most $k$ to positive matrices.
\end{enumerate}
\end{corollary}

In order to know the image of $\choi_\phi$ under ampliation maps,
we replace $\sigma$ by $\phi$ in {\sc Figure 1}
to consider the following diagrams:
$$
\begin{CD}
M_A @>\id_A>> M_A\\
@VV \phi V \\
M_B @>\psi^*>> M_A
\end{CD}
\qquad\qquad\qquad
\begin{CD}
M_A @>\psi>> M_B\\
@VV \phi V \\
M_B @>\id_B>> M_B
\end{CD}
$$
together with the identities
\begin{equation}\label{amp_ex_cp}
\begin{aligned}
\lan\phi,\psi\circ\sigma\ran
&=\lan\psi^*\circ\phi,\sigma\ran=\lan(\id_A\ot\psi^*)(\choi_\phi),\choi_\sigma\ran,\\
\lan\phi,\sigma\circ\psi\ran
&=\lan\phi\circ\psi^*,\sigma\ran=\lan(\psi\ot\id_B)(\choi_\phi),\choi_\sigma\ran.
\end{aligned}
\end{equation}
Therefore, we have
$$
\begin{aligned}
(K\circ\cp_{AA})^\circ
&=\{\phi\in H(M_A,M_B): (\id_A\ot\psi^*)(\choi_\phi)\in\pos_{AA} \ {\text{\rm for every}}\ \psi\in K\},\\
(\cp_{BB}\circ K)^\circ
&= \{\phi\in H(M_A,M_B: (\psi\ot\id_B)(\choi_\phi)\in\pos_{BB}\ {\text{\rm for every}}\ \psi\in K\}.
\end{aligned}
$$

\begin{theorem}\label{RMP2} \cite{gks}
For a closed convex cone $K\subset H(M_A,M_B)$, the following are equivalent:
\begin{enumerate}
\item[{\rm (i)}]
$K$ is a right mapping cone,
\item[{\rm (ii)}]
$\phi\in K^\circ$ if  and only if $(\id_A\ot\psi^*)(\choi_\phi)\ge 0$ for every $\psi\in K$.
\end{enumerate}
Furthermore, the following are also equivalent:
\begin{enumerate}
\item[{\rm (iii)}]
$K$ is a left mapping cone,
\item[{\rm (iv)}]
$\phi\in K^\circ$ if  and only if $(\psi\ot\id_B)(\choi_\phi)\ge 0$
for every $\psi\in K$.
\end{enumerate}
\end{theorem}

We take the convex cone $K=\superpos_k^\circ$ to get the following characterization of Schmidt numbers of states.
This gives rise to the separability criterion when $k=1$.

\begin{corollary}\label{sch-cri} \cite{{horo-1},{terhal-sghmidt}}
A state $\varrho\in M_A\ot M_B$ belongs to ${\mathcal S}_k$ if and only if
$(\id_A\ot\psi)(\varrho)\ge 0$ for every $k$-positive map $\psi: M_B\to M_A$ if and only if
$(\psi\ot\id_B)(\varrho)\ge 0$ for every $k$-positive map $\psi:M_A\to M_B$.
\end{corollary}

In order to know what happens to the image of $\choi_{\id_A}$ under the ampliation $\id_A\ot\phi$,
we put two identity maps in {\sc Figure 1} as follows:
$$
\begin{CD}
M_A @>\id_A>> M_A\\
@VV \id_A V \\
M_A @>\phi>> M_B
\end{CD}
$$
Then we get the identity  $(\id_A\ot\phi)(\choi_{\id_A})=\choi_\phi$, and have the following:

\begin{proposition}
A linear map $\phi:M_A\to M_B$ belongs to a convex cone $K$ if and only if $(\id_A\ot\phi)(\choi_{\id_A})$ belongs to
$\choi_K$.
\end{proposition}

It is worthwhile to collect results for $K^\circ=\superpos_k$ to get equivalent conditions
for $k$-superpositivity of a map $\phi:M_A\to M_B$, or Schmidt number of the corresponding state $\choi_\phi\in M_A\ot M_B$
as follows:

\begin{corollary}
For $\phi\in H(M_A,M_B)$, the following are equivalent:
\begin{enumerate}
\item[{\rm (i)}]
$\phi$ is $k$-superpositive, that is, $\phi=\sum_i\ad_{s_i}$ with $\rk s_i\le k$,
\item[{\rm (ii)}]
$\choi_\phi$ belongs to ${\mathcal S}_k$, that is,  has the Schmidt numbers $\le k$,
\item[{\rm (iii)}]
$\psi^*\circ\phi$ is completely positive for every $k$-positive map $\psi:M_A\to M_B$,
\item[{\rm (iv)}]
$\phi\circ\psi^*$ is completely positive for every $k$-positive map $\psi:M_A\to M_B$,
\item[{\rm (v)}]
$\id_A\ot\phi$ sends every state to a state with Schmidt number $\le k$,
\item[{\rm (vi)}]
$\phi\ot\id_B$ sends every $k$-blockpositive matrix to a positive matrix,
\item[{\rm (vii)}]
$(\id_A\ot\psi)(\choi_\phi)\ge 0$ for every $k$-positive map $\psi:M_B\to M_A$,
\item[{\rm (viii)}]
$(\psi\ot\id_B)(\choi_\phi)\ge 0$ for every $k$-positive map $\psi:M_A\to M_B$,
\item[{\rm (ix)}]
$(\phi\ot\psi)(\choi_{\id_A})\ge 0$ for every $k$-positive map $\psi:M_A\to M_B$,
\item[{\rm (x)}]
$(\phi^*\ot\psi^*)(\choi_{\id_B})\ge 0$ for every $k$-positive map $\psi:M_A\to M_B$,
\item[{\rm (xi)}]
$\phi\ot\psi$ is positive for every $k$-positive map $\psi:M_A\to M_B$.
\item[{\rm (xii)}]
$(\id_A\ot\phi)(\choi_{\id_A})$ is a state with Schmidt number $\le k$.
\end{enumerate}
\end{corollary}

If $s=|\xi\ran\lan\eta|$ is of rank one matrix, then we have
$$
\ad_s(a)=|\eta\ran\lan\xi|a|\xi\ran\lan\eta|=\lan \xi|a|\xi\ran\,  |\eta\ran\lan\eta|
=\lan a,|\bar\xi\ran\lan\bar\xi|\ran\, |\eta\ran\lan\eta|.
$$
Therefore, we see that $\phi$ is $1$-superpositive if and only if $\phi$ is of the following form
$$
\phi(a)=\sum_{k} \lan a, v_k\ran u_k
$$
with positive matrices $u_k$ and $v_k$. This is called the {\sl Holevo form} \cite{holovo98}.

In the paper \cite{hsrus}, a map $\phi$ was called  {\sl entanglement breaking}
when $\phi$ satisfies the the condition (v) with $k=1$,
and it was shown that conditions (i), (iii), (iv) and (xii) are equivalent to (v) together with the Holevo form.
On the other hand, a map was called in \cite{ando-04} superpositive
when its Choi matrix is separable.
Equivalent conditions (ii), (iii) and (v) for $k$-superpositive maps
were given in \cite{cw-EB}.

\section{Mapping cones arising from ampliation and factorizations}\label{constructio_MC}

A linear map $\phi:M_A\to M_B$ is called {\sl $k$-entanglement breaking} \cite{Christandl19}
if its ampliation map $\id_k\ot\phi:M_k\ot M_A\to M_k\ot M_B$ sends $(M_k\ot M_A)^+$ into ${\mathcal S}_1$.
The convex cone ${\mathbb E\mathbb B}_k$ of all $k$-entanglement breaking maps is also a mapping cone \cite{dms}.
Comparing with the mapping cones $\mathbb P_k$ and $\superpos_k$, the inclusion relations
can be summarized in {\sc Figure 2}.
Looking for mapping cones inside of the triangle in {\sc Figure 2}, it is natural to consider the convex cone
\begin{equation}\label{def-amp}
\amp_k[K,L]
:=\{\phi\in H(M_A,M_B): (\id_k\ot\phi)(\choi_{K})\subset \choi_{L}\},
\end{equation}
for $k=1,2,\dots$, closed convex cones $K\subset H(M_k,M_A)$ and $L\subset H(M_k,M_B)$.
Then we have
$$
\mathbb P_k =\amp_k[\cp,\cp],\qquad
\mathbb E\mathbb B_k=\amp_k[\cp,\superpos_1].
$$

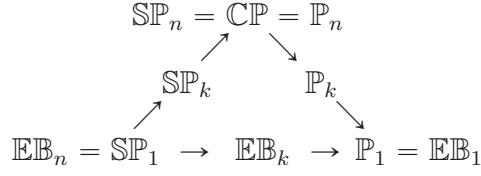
\begin{figure}\label{fig2}
\begin{center}
\setlength{\unitlength}{1 truecm}
\begin{picture}(5,2.5)
\put(-0.3,0){${\mathbb E\mathbb B}_n=\superpos_1  \ \rightarrow\ \, \mathbb E\mathbb B_k\  \rightarrow\, \mathbb P_1 =\mathbb E\mathbb B_1$}
\put(1.3,0.5){$\nearrow$}
\put(1.7,0.9){$\superpos_k$}
\put(2.2,1.4){$\nearrow$}
\put(1.3,1.8){$\superpos_n=\cp=\mathbb P_n$}
\put(3.1,1.4){$\searrow$}
\put(3.6,0.9){$\mathbb P_k$}
\put(4,0.5){$\searrow$}
\end{picture}
\end{center}
\caption{Inclusion relations when $M_A$ and $M_B$ are $n\times n$ matrices, where $X\to Y$ means $X\subset Y$.}
\end{figure}

The first identity of (\ref{amp_buyvyjj}) also shows the following relation
\begin{equation}\label{exam2}
(K\circ\cp_{AA})^\circ =\amp_a[\cp,K^\circ].
\end{equation}
Especially, we see that every mapping cone $K\subset H(M_A,M_B)$ can be expressed by
$$
K=(K^\circ\circ\cp_{AA})^{\circ}=\amp_a[\cp,K].
$$
This is nothing but Theorem \ref{RMP1} (ii). On other hand, the first identity of (\ref{amp_ex_cp})
may be written as $\lan\phi\circ\sigma^*,\psi\ran=\lan(\id_A\ot\psi^*)(\choi_\phi),\choi_\sigma\ran$, from which we have
\begin{equation}\label{mbygdcxsvf}
\amp_a[K,\cp]^*=(K\circ\cp_{AA})^\circ.
\end{equation}

As special cases of the identity (\ref{fund}), we have the identities
$$
(\id_k\ot\phi)(\choi_\sigma)=\choi_{\phi\circ\sigma},
\qquad
(\phi\ot\id_k)(\choi_{\sigma^*})=\choi_{\sigma^*\circ\phi^*},
$$
for $\sigma\in H(M_k,M_A)$ and $\phi\in H(M_A,M_B)$.
These identities imply the following:

\begin{proposition}\label{basic-def}
Suppose that $K$ and $L$ are closed convex cones in $H(M_k,M_A)$ and $H(M_k,M_B)$, respectively.
For $\phi\in H(M_A,M_B)$, the following are equivalent:
\begin{enumerate}
\item[(i)]
$\phi\in \amp_k[K,L]$, that is, $\id_k\ot\phi$ sends $\choi_{K}$ into $\choi_{L}$,
\item[(ii)]
For every $\sigma\in K$, we have $\phi\circ\sigma\in L$,
\item[(iii)]
For every $\sigma\in K$, we have $\sigma^*\circ\phi^*\in L^*$,
\item[(iv)]
$\phi\ot\id_k$ sends $\choi_{K^*}$ into  $\choi_{L^*}$.
\end{enumerate}
\end{proposition}

The equivalence between (i) and (iv) tells us that both the left and right ampliations basically give rise to
the same class of maps, in most cases. Note that ${K^*}\subset H(M_A, M_k)$ and ${L^*}\subset H(M_B, M_k)$.
It is clear that $\amp_k[K,L]$ is a closed convex cone.
The identities in (\ref{exam2}) and (\ref{mbygdcxsvf})
can be extended in much more general situations as follows:

\begin{proposition}\label{dual_amp}
For closed convex cones $K\subset H(M_k,M_A)$ and $L\subset
H(M_k,M_B)$, we have the following identities
$$
\amp_k[K,L]=(L^\circ\circ K^*)^{\circ},\qquad
\amp_k[K,L]^*=\amp_k[L^\circ,K^\circ].
$$
\end{proposition}

\begin{proof}
Suppose that $\phi\in\amp_k[K,L]$. Then for arbitrary given $\psi\in L^\circ$ and $\sigma\in K$, we have
$$
\lan \psi\circ\sigma^*,\phi\ran=\lan\psi,\phi\circ\sigma\ran\ge 0,
$$
since $\phi\circ\sigma\in L$. This implies that $\phi\in (L^\circ\circ K^*)^{\circ}$,
and we get 
$\amp_k[K,L]\subset (L^\circ\circ K^*)^{\circ}$.
For the reverse inclusion, suppose that $\phi\in (L^\circ\circ K^*)^{\circ}$ and $\sigma\in K$.
Then for every $\psi\in L^\circ$, we have
$$
\lan\phi\circ\sigma,\psi\ran=\lan \phi,\psi\circ\sigma^*\ran\ge 0,
$$
and so we have $\phi\circ\sigma\in L$. By Proposition \ref{basic-def} (ii), we have $\phi\in\amp_k[K,L]$.
This shows $(L^\circ \circ K^*)^\circ\subset \amp_k[K,L]$, and completes the proof of the first identity.
By the first identity, we have
$$
\amp_k[L^\circ,K^\circ]^*=
(K\circ L^{\circ*})^{\circ *}
=(K\circ L^{\circ*})^{*\circ}=(L^\circ\circ K^*)^\circ
=\amp_k[K,L],
$$
which shows the second identity.
\end{proof}

Motivated by Proposition \ref{dual_amp}, we consider the factorization properties.
A linear map $\phi:M_A\to M_B$ is called {\sl factorized through} $M_k$ when there exists $\sigma:M_A\to M_k$
and $\tau:M_k\to M_B$ such that $\phi=\tau\circ\sigma$.
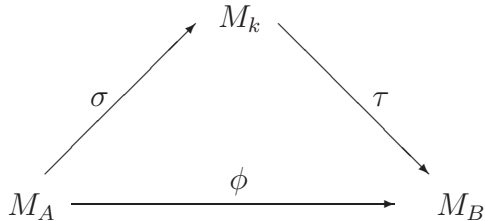
\begin{center}
\setlength{\unitlength}{0.05 truecm}
\begin{picture}(120,65)
\put (0,0){$M_A$}
\put (17,3){\vector (1,0){86}}
\put (114,0){$M_B$}
\put (59,7){$\phi$}
\put (10,11){\vector(1,1){40}}
\put (56,50){$M_k$}
\put (72,51){\vector (1,-1){40}}
\put (97,29){$\tau$}
\put (22,29){$\sigma$}
\end{picture}
\end{center}
For a given natural number $k=1,2,\dots$, and closed convex cones $K\subset H(M_k,M_A)$, $L\subset H(M_k,M_B)$, we define
$$
\mathbb F_k[K,L]
=\{\tau\circ\sigma^*:\sigma\in K,\ \tau\in L\}^{\circ\circ}.
$$
So, maps in $\mathbb F_k[K,L]$ are finite sums of linear maps which are factorized through $M_k$ as
compositions $\tau\circ\sigma^*$ with $\sigma\in K$ and $\tau\in L$.
By Proposition \ref{dual_amp}, we have the following:

\begin{theorem}\label{the_dual_amp}
For a given natural number $k=1,2,\dots$, and closed convex cones $K\subset H(M_k,M_A)$, $L\subset H(M_k,M_B)$, we have
$$
\amp_k[K,L]^\circ=\mathbb F_k[K,L^\circ].
$$
\end{theorem}

It is easily seen that the relation $\mathbb F_k[\cp,\cp]=\superpos_k$ holds. This recovers
$\superpos_k^\circ=\amp_k[\cp,\cp]=\mathbb P_k$ by Theorem \ref{the_dual_amp}.
We also have $\mathbb E\mathbb B_k^\circ=\mathbb F_k[\cp,\mathbb P_1]$, which recovers \cite[Theorem 3.13]{dms}.

Now, we look for conditions with which $\amp_k[K,L]$ and/or $\mathbb F_k[K,L]$ are mapping cones.
To do this, we use the condition (ii) of Proposition \ref{basic-def}.

\begin{lemma}
For closed convex cones $K\subset H(M_k,M_A)$ and $L\subset H(M_k,M_B)$,
we have the following:
\begin{enumerate}
\item[(i)]
If $\cp_B\circ L\subset L$ then $\cp_B\circ \amp_k[K,L]\subset \amp_k[K,L]$.
\item[(ii)]
If $\cp_A\circ K\subset K$ then $\amp_k[K,L]\circ\cp_A\subset \amp_k[K,L]$.
\end{enumerate}
\end{lemma}

\begin{proof}
Suppose that $\cp_B\circ L\subset L$ holds, and take $\tau\in\cp_B$ and $\phi\in \amp_k[K,L]$.
For any $\sigma\in K$, we have $\phi\circ\sigma\in L$ by the condition (ii) of Proposition \ref{basic-def}. This implies that
$(\tau\circ\phi)\circ\sigma=\tau\circ(\phi\circ\sigma)\in \cp_B\circ L\subset L$. Therefore, we have $\tau\circ\phi\in \amp_k[K,L]$
by (ii) of Proposition \ref{basic-def}.

For (ii), suppose $\cp_A\circ K\subset K$, and take $\phi\in \amp_k[K,L]$ and $\tau\in\cp_A$.
Then, for every $\sigma\in K$, we have $\tau\circ\sigma\in \cp_A\circ K\subset K$, which implies
$(\phi\circ\tau)\circ\sigma=\phi\circ(\tau\circ\sigma)\in L$. Therefore, we have $\phi\circ\tau\in \amp_k[K,L]$.
\end{proof}

If both convex cones $K\subset H(M_k,M_A)$ and $L\subset H(M_k,M_B)$ consist of positive maps then it is clear that
$\mathbb F_k(K,L)$ also consists of positive maps. Since the dual cone of a mapping cone is also a mapping cone, we have the following:

\begin{theorem}
Suppose that both $K$ and $L$ are left mapping cones.
Then the convex cones $\amp_k[K,L]$ and $\mathbb F_k[K,L^\circ]$ are mapping cones.
\end{theorem}

Therefore, we see that $\mathbb E\mathbb B_k$ is a mapping cone, as it was shown in \cite{dms}.
The relation (\ref{exam2}) recovers the fact that $(K\circ\cp)^{\circ\circ}$ is a mapping cone whenever $K$ is a left mapping cone.
It is clear that $\mathbb F_k[K,L]$ is increasing with respect to the both variables $K$ and $L$.
On the other hand, $\amp_k[K,L]$ are increasing with respect to the variable $L$ and decreasing with respect to $K$.
We consider the case of $K=\cp$ to get various mapping cones.

We fix $M_A=M_B=M_n$, $K=\cp$ and consider the cases of $L=\superpos_\ell$ or $L=\mathbb P_\ell$ with $\ell=1,2,\dots,n$.
If $\ell\ge k$, then we have $\superpos_k=\superpos_\ell=\cp=\mathbb P_\ell=\mathbb P_k$
in $H(M_k,M_n)$. Therefore, all the mapping cones in the following inclusions
$$
\amp_k[\cp,\superpos_k]\subset
\amp_k[\cp,\superpos_\ell]\subset
\amp_k[\cp,\cp]\subset
\amp_k[\cp,\mathbb P_\ell]\subset
\amp_k[\cp,\mathbb P_k]
$$
coincide with $\mathbb P_k$.
We also see that all the inclusions
$$
\mathbb P_k=\amp_n[\cp,\mathbb P_k]\subset \amp_\ell[\cp,\mathbb P_k]\subset\amp_k[\cp,\mathbb P_k]=\mathbb P_k
$$
become identities. In short, we have
$\amp_k[\cp,\mathbb P_\ell]=\mathbb P_{\min\{k,\ell\}}$.
The inclusion relations among all the remaining mapping cones are summarized in {\sc Figure 3}
with the notations
$$
\begin{aligned}
K_{k,\ell}&:=\amp_k[\cp,\superpos_\ell]
=\{\phi\in H(M_A,M_B): (\id_k\ot\phi)({\mathcal P})\subset {\mathcal S}_\ell\},\\
\end{aligned}
$$
for $k,\ell=1,2,\dots,n$ with $k\ge \ell$.

\begin{figure}\label{fig3}
\begin{center}
\setlength{\unitlength}{1 truecm}
\begin{picture}(12,6)
\put(-0.3,1){$\superpos_1=K_{n,1}\ \rightarrow\ K_{n-1,1}\ \rightarrow$}
\put(5.6,1){$\rightarrow\ K_{3,1}\ \ \rightarrow\ K_{2,1}\ \ \rightarrow\ K_{1,1}=\mathbb P_1$}
\put(0.8,2){$\superpos_2=K_{n,2} \ \rightarrow\ K_{n-1,2}\ \rightarrow$}
\put(6.3,2){$\rightarrow\ K_{3,2}\ \rightarrow\ K_{2,2}=\mathbb P_2$}
\put(2.2,4){$\superpos_{n-1}=K_{n,n-1}$}
\put(6.2,4){$K_{n-1,n-1}=\mathbb P_{n-1}$}
\put(3.9,5.4){$\superpos_n=K_{n,n}=\mathbb P_n=\cp$}
\put(1.5,1.5){$\nearrow$}
\put(2.5,2.5){$\nearrow$}
\put(3.4,3.4){$\nearrow$}
\put(4.7,4.7){$\nearrow$}
\put(6,4.7){$\searrow$}
\put(7.3,3.4){$\searrow$}
\put(8.2,2.5){$\searrow$}
\put(9.2,1.5){$\searrow$}
\put(3.5,1.5){$\nearrow$}
\put(4.5,2.5){$\nearrow$}
\put(6.6,1.5){$\nearrow$}
\put(8.6,1.5){$\nearrow$}
\put(7.6,2.5){$\nearrow$}
\put(1,0){$\mathbb E\mathbb B_n$}
\put(1.2,0.5){$\shortparallel$}
\put(2.8,0){$\mathbb E\mathbb B_{n-1}$}
\put(3.2,0.5){$\shortparallel$}
\put(6.3,0){$\mathbb E\mathbb B_3$}
\put(6.5,0.5){$\shortparallel$}
\put(8.3,0){$\mathbb E\mathbb B_2$}
\put(8.5,0.5){$\shortparallel$}
\put(10.3,0){$\mathbb E\mathbb B_1$}
\put(10.5,0.5){$\shortparallel$}
\put(4.8,1){$\cdots$}
\put(5.6,2){$\cdots$}
\put(5.4,4){$\rightarrow$}
\put(3,2.8){$\cdot$}\put(3.1,2.9){$\cdot$}\put(3.2,3){$\cdot$}
\put(8,2.8){$\cdot$}\put(7.9,2.9){$\cdot$}\put(7.8,3){$\cdot$}
\end{picture}
\end{center}
\caption{}
\end{figure}

In order to distinguish $k$-positivities, Tomiyama \cite{tom_85}
considered the linear map $\phi_\lambda:M_n\to M_n$ given by
$$
\phi_\lambda(x)=\lambda\tr(x)I_n-x,\qquad x\in M_n,
$$
with the parameter $\lambda\ge 1$, and showed that $\phi_\lambda$ is $k$-positive if and only if $\lambda\ge k$.
See also \cite{tomiyama-83}. The map $\phi_{n-1}$ was the example of Choi \cite{choi72}
to distinguish $n$-positivity and $(n-1)$-positivity.
It was shown in \cite{{chen_chit_2020},{dms}} that $\phi_k$
is even $k$-entanglement breaking as well as $k$-positive. Therefore, we have
$\mathbb E\mathbb B_k\nsubset \mathbb P_{k+1}$.
The following tells us that there is no more possible inclusion relations in {\sc Figure 3}.

\begin{proposition}
Suppose that $k,\ell,p,q=1,2,\dots,n$ with $k\ge \ell$ and $p\ge q$. Then $K_{k,\ell}\subset K_{p,q}$ if and only if
the following two conditions are satisfied:
\begin{enumerate}
\item[(i)]
$k\ge p$,
\item[(ii)]
$\ell\le q$ or $\ell>q=p$.
\end{enumerate}
\end{proposition}

\begin{proof}
The \lq if\rq\ part is clear from {\sc Figure 3}. Now, we suppose that $K_{k,\ell}\subset K_{p,q}$ holds.
If $k<p$ then we have
$\mathbb E\mathbb B_k=K_{k,1}\subset K_{k,\ell}\subset K_{p,q}\subset K_{p,p}=\mathbb P_p\subset \mathbb P_{k+1}$,
to get a contradiction. This shows that $k\ge p$.

To prove the condition (ii), we take $s=\sum_{i=1}^\ell e_{i,i}\in
M_n$, and $\varrho=\sum_{i,j=1}^p e_{i,j}\ot e_{i,j}\in M_p \ot M_n$.
Then we have $\ad_s\in \superpos_\ell=K_{n,\ell}\subset
K_{k,\ell}\subset K_{p,q}$ and $\varrho$ is positive in $M_p\ot M_n$. Therefore, we see that
$$
\sum_{i,j=1}^{\min\{\ell,p\}} e_{i,j}\ot e_{i,j}=(\id_p\ot \ad_s)(\varrho)\in M_p\ot M_n
$$
belongs to ${\mathcal S}_q\subset M_p\ot M_n$. This implies $\min\{\ell,p\}\le q$, and
gives rise to the condition (ii), since $p\ge q$.
\end{proof}

In the case of $n=3$, we note that $\phi$ belongs to the mapping cone $K_{2,1}=\mathbb E\mathbb B_2$
if and only if the map $\id_2\ot\phi:M_2\ot M_3\to M_2\ot M_3$ send ${\mathcal P}$ into
${\mathcal S}_1$. Since ${\mathcal S}_1=\ppt$ in $M_2\ot M_3$ \cite{{p-horo},{woronowicz}},
we see \cite{Christandl19} that $\phi\in K_{2,1}$
if and only if $\phi$ is both $2$-positive and $2$-copositive,
that is, $\phi\circ\ttt$ is $2$-positive.
See \cite{cho-kye-lee} for parameterized examples of linear maps between $M_3$ which are
both $2$-positive and $2$-copositive but not completely positive.

\section{PPT square conjecture}\label{PPT-sec}

The PPT square conjecture claims that the composition of two PPT maps is $1$-superpositive, that is, entanglement breaking,
as it was proposed by Christandl in \cite{ppt}. This conjecture which can be expressed by
$$
\pptmap\circ\pptmap\subset\superpos_1,
$$
is supported by several results:
If $\phi$ is a unital or trace preserving PPT map then
$d(\phi^k, \superpos_1)$ tends to $0$ as $k\to\infty$ \cite{kennedy17};
if $\phi$ is a unital PPT map then $\phi^n\in \superpos_1$ for a positive integer $n$ \cite{rahaman18};
the conjecture is true when $\phi$ is a PPT map between $3\times 3$ matrices
\cite{{chen_yany_tang},{Christandl19}}. See also \cite{collins_PPT}.
It was also shown in \cite{Christandl19} that the conjecture (\ref{dec-con}) implies that
a finite iteration of a PPT map is entanglement breaking.

Equivalent claims and possibility to find counterexamples also have been discussed. It was shown in
\cite{m-h_2018} that the existence of a nontrivial tensor-stable positive map implies the negation of the conjecture.
Recall that a map $\phi$ is called {\sl tensor-stable positive} \cite{M-HRW_2016}
if $\phi^{\otimes n}$ is positive for every $n=1,2,\dots$, and
note that completely positive or completely copositive maps are trivially tensor-stable positive.
It was also shown in \cite{Christandl19} that the conjecture is true if and only if
$\mathbb P_1\circ\pptmap\subset\dec$ if and only if
$\pptmap\ot  \pptmap (\pos)\subset \mathcal S_1$ holds.
In \cite{gks}, several equivalent claims to the PPT square conjecture are found through composition as follows:
\begin{proposition}\label{fhbukkm}\cite{gks}
The following statements are equivalent:
\begin{enumerate}
\item[(i)]
$\pptmap \circ \pptmap \subset \superpos_1$.
\item[(ii)]
$\pptmap \circ \mathbb P_1 \subset \dec$.
\item[(iii)]
$\mathbb P_1 \circ \pptmap \subset \dec$.
\item[(iv)]
$\pptmap \circ \cp \circ \pptmap \subset \superpos_1$.
\item[(v)]
$\pptmap \circ \dec \circ \pptmap \subset \superpos_1$.
\item[(vi)]
$\pptmap \circ \mathbb P_1 \circ \pptmap \subset \cp$.
\item[(vii)]
$\pptmap \circ \mathbb P_1 \circ \pptmap \subset \pptmap$.
\end{enumerate}
\end{proposition}
Here, equivalence between (i) and (iv) follows, since $\pptmap$ is a mapping cone. The other claims are
easy consequences of the identities in (\ref{pair_dual}).
For example, the identity
$$
\lan\phi_1\circ\sigma\circ\phi_2,\psi\ran=\lan\sigma,\phi_1^*\circ\psi\circ\phi_2^*\ran
$$
with $\psi\in\mathbb P_1$, $\phi_i\in\pptmap$ and $\sigma\in\cp$ proves (iv) $\Longleftrightarrow$ (vi).
Claims in Proposition \ref{fhbukkm} may be translated into those in terms of tensor products.
By the identity (\ref{fund}), we have the relation
\begin{equation}\label{comp1}
\choi_{K_3\circ K_2\circ K_1^*}=(K_1\ot K_3)(\choi_{K_2})
\end{equation}
for any convex cones $K_1$, $K_2$ and $K_3$, whenever the above
expression is meaningful. Furthermore, we also have
\begin{equation}\label{comp2}
\choi_{K_3\circ K_1^*}=(K_1\ot K_3)(\choi_{\id_A})
\end{equation}
whenever $K_1,K_3\subset H(M_A,M_B)$. Therefore, we see that
$\pptmap\circ\pptmap\subset \superpos_1$ if and only if
$(\pptmap\ot\pptmap)(\choi_{\id_A})\subset \choi_{\superpos_1}$. In this way,
we use (\ref{comp1}) and (\ref{comp2}) to translate the statements
Proposition \ref{fhbukkm} to those in terms of tensor products of
linear maps.

\begin{proposition}\label{ppt1}
The following are equivalent to the PPT square conjecture:
\begin{enumerate}
\item[(i)]
$(\pptmap\ot\pptmap)(\choi_{\id_A})\subset {\mathcal S}_1$,
\item[(ii)]
$(\mathbb P_1\ot \pptmap)(\choi_{\id_A})\subset \decmat$,
\item[(iii)]
$(\pptmap\ot \mathbb P_1)(\choi_{\id_A})\subset \decmat$,
\item[(iv)]
$(\pptmap\ot\pptmap)(\pos)\subset {\mathcal S}_1$,
\item[(v)]
$(\pptmap\ot\pptmap)(\decmat)\subset {\mathcal S}_1$,
\item[(vi)]
$(\pptmap\ot\pptmap)(\blockpos_1)\subset {\mathcal P}$,
\item[(vii)]
$(\pptmap\ot\pptmap)(\blockpos_1)\subset \ppt$.
\end{enumerate}
\end{proposition}
The equivalent statement (iv) observed in \cite{Christandl19} claims that if $\phi$ and $\psi$ are PPT maps
then $\phi\ot\psi$ is an entanglement annihilating map. This is true for maps between $3\times 3$ matrices,
since the PPT square conjecture is true in this case \cite{{chen_yany_tang},{Christandl19}}.

Suppose that $\phi_i\in H(M_{A_i},M_{B_i})$ for $i=1,2$, $\phi_3\in H(M_{A_1},M_{A_2})$ and
$\phi_4\in H(M_{B_1},M_{B_2})$. Then we have the identity
$$
\begin{aligned}
\lan(\phi_1\ot \phi_2)(\choi_{\phi_3}),\choi_{\phi_4}\ran
&=\lan\phi_2\circ\phi_3\circ\phi_1^*,\phi_4\ran\\
&=\lan\phi_4\circ\phi_1\circ\phi_3^*,\phi_2\ran\\
&=\lan (\phi_3\ot\phi_4)(\choi_{\phi_1}),\choi_{\phi_2}\ran.
\end{aligned}
$$
Hence, we see that the following equivalence
$$
(K_1\ot K_2)(\choi_{K_3})\subset \choi_{K_4}
\ \Longleftrightarrow\
(K_3\ot K_4^\circ)(\choi_{K_1})\subset \choi_{K_2}^\circ
$$
holds, since the identity holds for $\phi_i\in K_i$ for $i=1,2,3$ and $\phi_4\in K_4^\circ$.
Therefore, we have the following equivalent statements to the PPT square conjecture.

\begin{proposition}\label{ppt2}
The following are equivalent to the PPT square conjecture:
\begin{enumerate}
\item[(i)]
$(\id_A\ot{\mathbb P}_1)(\ppt)\subset \decmat$,
\item[(ii)]
$(\id_A\ot \pptmap)(\blockpos_1)\subset \decmat$,
\item[(iii)]
$(\id_A\ot \pptmap)(\ppt)\subset {\mathcal S}_1$,
\item[(iv)]
$(\cp\ot{\mathbb P}_1)(\ppt)\subset \decmat$,
\item[(v)]
$(\dec\ot{\mathbb P}_1)(\ppt)\subset \decmat$,
\item[(vi)]
$(\mathbb P_1\ot\cp)(\ppt)\subset \decmat$,
\item[(vii)]
$(\mathbb P_1\ot\dec)(\ppt)\subset \decmat$.
\end{enumerate}
\end{proposition}

We note that the maps in Proposition \ref{ppt2} are linear maps from $M_A\ot M_B$ into itself, while
maps in Proposition \ref{ppt1} were maps from $M_A\ot M_A$ into $M_B\ot M_B$.
Equivalent claims of the conjecture through ampliation as in (i), (ii) and (iii) of Proposition \ref{ppt2}
can be also obtained from Theorem \ref{the_dual_amp}.
In fact, $K_1\circ K_2\subset K_3$ holds if and only if
${\mathbb F}_a(K_1,K_2)\subset K_3$ if and only if
$K_3^\circ\subset\amp_a[K_1,K_2^\circ]$ by Theorem \ref{the_dual_amp}.
Therefore, we have
$$
K_1\circ K_2\subset K_3
\ \Longleftrightarrow\
(\id_A\ot K_3^\circ)(\choi_{K_1})\subset \choi_{K_2^\circ}.
$$
For an example, we have $\pptmap\circ\pptmap\subset {\mathcal S}_1$ if and only if
$(\id\ot \mathbb P_1)(\ppt)\subset \decmat$.
In this way, we see that the first three claims in Proposition \ref{fhbukkm}
are equivalent to those in Proposition \ref{ppt2}.

As for linear maps between $3\times 3$ matrices, it was shown in \cite{aubrun_mul} recently
that the following relation
$$
\pptmap\ot  \pptmap (\blockpos_1)\subset \mathcal S_1
$$
holds. Comparing with Proposition \ref{ppt1} (iv), (v) (vi) and (vii), this is seemingly stronger than the PPT square conjecture.
We also note that this is equivalent to the claims
$$
\pptmap\circ \mathbb P_1 \circ \pptmap \subset \superpos_1,\quad {\text{\rm and}}\quad
(\mathbb P_1\ot\mathbb P_1)(\ppt)\subset\decmat,
$$
respectively.
It would be interesting to know if these are equivalent to the PPT square conjecture in general.

\section{Discussion}\label{discussion}

In this article, we began with the concrete examples of positive maps of the form $x\mapsto s^*xs$
to introduce the convex cones in (\ref{diagram}) which appear as key notions in the current quantum information theory.
Main tools were Choi matrices and duality arising from the bilinear pairing.
Such convex cones are mapping cones whose dual cones can be described in terms composition
and tensor products of linear maps.

Because $\superpos_k$ is singly generated as a mapping cone,
elements of the dual cone can be described by the tensor product with a single map $\id_k$.
Especially, we have $\phi\in\superpos_k^\circ$ if and only if $\id_k\ot\phi\in \superpos_1^\circ$,
and see that the ampliation connects $\superpos_k^\circ$ and $\superpos_1^\circ$.
A natural question arises: Does there exist an operation which connects $\superpos_k$ and $\superpos_1$?
Such an operation would be very useful as criteria for $k$-superpositivity of maps and Schmidt numbers
of bi-partite states, because various known separability criteria may be used for this purpose.

The identity (\ref{fund}) played essential roles through the whole discussion. It is also natural
to look for analogous identities for arbitrary iterations of tensor products and compositions.
Suppose that $\phi_i: M_{A_i}\to M_{B_i}$ is a linear map for $i=1,2,\dots,n$. Then we have the linear map
$$
\phi_1\ot\cdots\ot\phi_n: M_{A_1}\ot\cdots \ot M_{A_n}\to M_{B_1}\ot\cdots\ot M_{B_n}.
$$
We can take a nontrivial bi-partition $S\sqcup T=\{1,2,\dots,n\}$ of systems, then matrices in the domain can be
dealt as the Choi matrix of a linear map from $\bigotimes_{i\in S} M_{A_i}$ to $\bigotimes_{i\in T} M_{A_i}$.
See \cite{han_kye_optimal} for more details.
Then we may apply (\ref{fund}) to get the identity.
Alternatively, one may use the Choi matrices of multi-linear maps from $M_{A_1}\times\cdots \times M_{A_{n-1}}$
to $M_{A_n}$ as it was considered in \cite{kye_multi_dual}.
Then we use the similar calculation as in the proof of (\ref{fund}) to get an identity.
In both methods, we cannot obtain identities involving arbitrary iterations of a map by composition.
Such an identity may be useful to deal with tensor-stable positive maps.


\begin{thebibliography}{99}

\bibitem{ando-04}
T. Ando,
\it Cones and norms in the tensor product of matrix spaces,
\rm Linear Alg. Appl. \bf 379 \rm (2004), 3--41.

\bibitem{aubrun_mul}
G. Aubrun and A. M\" uller-Hermes, \it Annihilation entanglement
between cones, \rm Preprint, arXiv 2110.11825.

\bibitem{book_au_Sz}
G. Aubrun and S. J. Szarek,
\lq\lq Alice and Bob meet Banach:
The Interface of Asymptotic Geometric Analysis and Quantum Information Theory\rq\rq,
Math. Surveys Monog. Vol 223, Amer. Math. Soc., 2017.

\bibitem{bhat-osaka-2020}
B. V. R. Bhat and H. Osaka,
\it A factorization property of positive maps on $C^*$-algebras,
\rm Intern. J. Quantum Inform. {\bf 8} (2020), 52.

\bibitem{criello}
D. Cariello,
\it Inequalities for the Schmidt number of bipartite states,
\rm Lett. Math. Phis. {\bf 110} (2020), 827--833.

\bibitem{chen_chit_2020}
S. Chen and E. Chitambar,
\it Entanglement-breaking superchannels,
\rm Quantum {\bf 4} (2020), 299.

\bibitem{chen_yany_tang-2017}
L. Chen, Y. Yang and W.-S. Tang,
\it Schmidt number of bipartite and multipartite states under local projections
\rm Quantum Inf. Process {\bf 16} (2017), 75.



\bibitem{chen_yany_tang}
L. Chen, Y. Yang and W.-S. Tang,
\it Positive-partial-transpose
square conjecture for $n = 3$, \rm Phys. Rev. A {\bf 99} (2019),
012337.

\bibitem{cho-kye-lee}
S.-J. Cho, S.-H. Kye and S. G. Lee, \it Generalized Choi maps in
$3$-dimensional matrix algebras, \rm Linear Alg. Appl. \bf 171 \rm
(1992), 213--224.

\bibitem{choi72}
M.-D. Choi,
\it Positive linear maps on $C^*$-algebras,
\rm Canad. Math. J. \bf 24 \rm (1972), 520--529.

\bibitem{choi75-10}
M.-D. Choi, \it Completely positive linear maps on complex matrices,
\rm Linear Alg. Appl. \bf 10 \rm (1975), 285--290.

\bibitem{choi-ppt}
M.-D. Choi,
\it Positive linear maps,
\rm Operator Algebras and Applications (Kingston, 1980), pp. 583--590,
Proc. Sympos. Pure Math. Vol 38. Part 2, Amer. Math. Soc., 1982.

\bibitem{Christandl19}
M. Christandl, A. M\" uller-Hermes and M. M. Wolf, \it When Do
Composed Maps Become Entanglement Breaking? \rm Ann. Henri Poincar\'
e {\bf 20} (2019), 2295--2322.

\bibitem{cw-EB}
D. Chru\'{s}ci\'{n}ski and A. Kossakowski,
\it On Partially Entanglement Breaking Channels,
\rm Open Sys. Information Dyn. {\bf 13} (2006), 17--26.

\bibitem{collins_PPT}
B. Collins, Z. Yin and P. Zhong, \it The PPT-square conjecture holds
generically for some classes of independent states, \rm J. Phys. A:
Math. Theor. {\bf 51} (2018), 425301.

\bibitem{dePillis}
J. de Pillis,
\it Linear transformations which preserve Hermitian and positive semidefinite operators,
\rm Pacific J. Math. {\bf 23} (1967), 129--137.

\bibitem{dms}
R. Devendra, N. Mallick and K. Sumesh,
\it Mapping cone of $k$-entanglement breaking maps,
\rm preprint. arXiv 2105.14991.

\bibitem{eom-kye}
M.-H. Eom and S.-H. Kye,
\it Duality for positive linear maps in matrix algebras,
\rm Math. Scand. \bf 86 \rm (2000), 130--142.

\bibitem{gks}
M. Girard, S.-H. Kye and E. St\o rmer,
\it Convex cones in mapping spaces between matrix algebras,
\rm Linear Algebra Appl. {\bf 608} (2021), 248--269.

\bibitem{han_kye_tri}
K. H. Han and S.-H, Kye,
\it Various notions of positivity for bi-linear maps and applications to tri-partite entanglement,
\rm J. Math. Phys. {\bf 57} (2016), 015205.

\bibitem{han_kye_optimal}
K. H. Han and S.-H, Kye,
\it Construction of multi-qubit optimal genuine entanglement witnesses,
\rm J. Phys. A: Math. Theor. {\bf 49} (2016), 175303.


\bibitem{holovo98}
A. S. Holevo,
\it Quantum coding theorems,
\rm Russian Math. Surveys {\bf 53} (1998), 1295--1331.

\bibitem{holevo_2011}
A. S. Holevo,
\it Entropy gain and the Choi--Jamiolkowski correspondence for infinite dimensional quantum evolutions,
\rm Theor. Math. Phys. {\bf 166} (2011), 123--138.

\bibitem{holevo_2011_a}
A. S. Holevo,
\it The Choi--Jamiolkowski forms of quantum Gaussian channels,
\rm J. Math. Phys. {\bf 52} (2011), 042202.

\bibitem{horo-1}
M. Horodecki, P. Horodecki and R. Horodecki,
\it Separability of mixed states: necessary and sufficient conditions,
\rm Phys. Lett. A \bf 223 \rm (1996), 1--8.

\bibitem{hsrus}
M. Horodecki, P. W. Shor and M. B. Ruskai,
\it Entanglement braking channels,
\rm Rev. Math. Phys. \bf 15 \rm (2003), 629--641.

\bibitem{p-horo}
P. Horodecki,
\it Separability criterion and inseparable mixed states with positive partial transposition,
\rm Phys. Lett. A \bf 232 \rm (1997),  \rm 333--339.

\bibitem{hllm}
M. Huber, L. Lami, C. Lancien and A. M\' uller-Hermes,
\it High-Dimensional Entanglement in States with Positive Partial Transposition,
\rm Phys. Rev. Lett. {\bf 121} (2018), 200503.



\bibitem{jam_72}
A. Jamio\l kowski,
\it Linear transformations which preserve trace and positive semidefinite operators,
\rm Rep. Math. Phys. {\bf 3} (1972), 275--278.

\bibitem{kennedy17}
M. Kennedy, N. Manor and V. I. Paulsen,
\it Composition of PPT Maps,
\rm Quantum Inform. Comput. {\bf 18} (2018), 0472--0480.

\bibitem{kraus}
K. Kraus,
\it Operations and effects in the Hilbert space
formulation of quantum theory,
\rm Foundations of quantum mechanics and ordered
linear spaces (Marburg, 1973), pp. 206--229. Lecture Notes
in Phys., Vol. 29,  Springer, 1974.

\bibitem{kye_rev}
S.-H. Kye,
\it Facial structures for various notions of positivity and applications to the theory of entanglement,
\rm Rev. Math. Phys. {\bf 25} (2013), 1330002.

\bibitem{kye_multi_dual}
S.-H. Kye,
\it Three-qubit entanglement witnesses with the full spanning properties,
\rm J. Phys. A: Math. Theor. {\bf 48} (2015), 235303.

\bibitem{kye-choi_mat}
S.-H. Kye,
\it Choi matrices revisited,
\rm J. Math. Phys. {\bf 63} (2022), 092202.

\bibitem{li_du}
Y. Li and H.-K. Du,
\it Interpolations of entanglement breaking channels and equivalent conditions for completely positive maps,
\rm J. Funct. Anal. {\bf 268} (2015), 3566--3599.

\bibitem{Magajna_2021}
B. Magajna,
\it Cones of completely bounded maps,
\rm Positivity {\bf 25} (2021), 1--29.

\bibitem{marcin_exp}
M. Marciniak,
\it Rank properties of exposed positive maps,
\rm Linear Multilinear Alg. \bf 61 \rm (2013), 970--975.


\bibitem{m-h_2018}
A. M\" uller-Hermes, \it Decomposability of linear maps under tensor
powers, \rm J. Math. Phys. {\bf 59} (2018), 102203.

\bibitem{M-HRW_2016}
A. M\" uller-Hermes, D. Reeb and M. M. Wolf, \it Positivity of
linear maps under tensor powers, \rm J. Math. Phys. {\bf 57} (2016),
015202.

\bibitem{Paulsen_Shultz}
V. I. Paulsen and F. Shultz,
\it Complete positivity of the map from a basis to its dual basis,
\rm J. Math. Phys. {\bf 54} (2013), 072201.

\bibitem{peres}
A. Peres,
\it Separability Criterion for Density Matrices,
\rm Phys. Rev. Lett. \bf 77 \rm (1996), 1413--1415.

\bibitem{rahaman18}
M. Rahaman, S. Jaques and V. I. Paulsen,
\it Eventually entanglement breaking maps,
\rm J. Math. Phys. {\bf 59} (2018), 062201.

\bibitem{ppt}
M. B. Ruskai, M. Junge, D. Kribs, P. Hayden, A. Winter, \it Operator
structures in quantum information theory, \rm Final Report, Banff
International Research Station (2012).

\bibitem{sbl}
A. Sanpera, D. Bru\ss\ and M. Lewenstein,
\it Schmidt number witnesses and bound entanglement,
\rm Phys. Rev. A \bf 63 \rm (2001), 050301.

\bibitem{shor}
P. W. Shor,
\it Additivity of the classical capacity of entanglement-breaking quantum channels,
\rm J. Math. Phys. {\bf 43} (2002), 4334--4340.

\bibitem{sko-laa}
\L . Skowronek,
\it Cones with a mapping cone symmetry in the finite-dimensional case,
\rm Linear Algebra Appl. {\bf 435} (2011), 361--370.

\bibitem{ssz}
\L. Skowronek, E. St\o rmer, and K. Zyczkowski,
\it Cones of positive maps and their duality relations,
\rm J. Math. Phys. {\bf 50}, (2009), 062106.

\bibitem{stine}
W. M. Stinespring,
\it Positive functions on $C^*$-algebras,
\rm Proc. Amer. Math. Soc. {\bf 6} (1955), 211--216.

\bibitem{stormer}
E. St\o rmer,
\it Positive linear maps of operator algebras,
\rm Acta Math. \bf 110 \rm (1963), 233--278.

\bibitem{stormer82}
E. St\o rmer,
\it Decomposable positive maps on $C^*$-algebras,
\rm Proc. Amer. Math. Soc. \bf 86 \rm (1982), 402--404.

\bibitem{stormer-dual}
E. St\o rmer,
\it Extension of positive maps into $B(\mathcal H)$,
\rm J. Funct. Anal. \bf 66 \rm (1986), 235-254.


\bibitem{stormer_scand_2012}
E. St\o rmer, \it Tensor products of positive maps of matrix
algebras, \rm Math. Scand. {\bf 111} (2012), 5--11.

\bibitem{stormer_book}
E. St\o rmer,
Positive Linear Maps of Operator Algebras,
Springer-Verlag, 2013.

\bibitem{stormer_choi_mat}
E. St\o rmer,
\it The analogue of Choi matrices for a class of linear maps on Von Neumann algebras,
\rm Inern. J. Math. {\bf 26} (2016), 1550018.


\bibitem{tomiyama-83}
K. Tanahashi and J. Tomiyama, \it On the geometry of positive maps
in matrix algebras, \rm Math. Z. \bf 184 \rm (1983), 101--108.

\bibitem{terhal-sghmidt}
B. M. Terhal and P. Horodecki,
\it A Schmidt number for density matrices
\rm Phys. Rev. A \bf 61 \rm (2000), 040301.

\bibitem{tom_85}
J. Tomiyama, \it On the geometry of positive maps in matrix
algebras. II, \rm  Linear Alg. Appl.  \bf 69 \rm (1985), 169--177.

\bibitem{Werner-1989}
R. F. Werner,
\it Quantum states with {E}instein-{P}odolsky-{R}osen correlations admitting a hidden-variable model,
\rm Phys. Rev. A, {\bf 40} (1989), 4277--4281.

\bibitem{woronowicz}
S. L. Woronowicz,
\it Positive maps of low dimensional matrix algebras,
\rm Rep. Math. Phys. \bf 10 \rm (1976), 165--183.

\bibitem{ylt}
Y. Yang, D. H. Leung and W.-S. Tang,
\it All $2$-positive linear maps from $M_3(\mathbb C)$ to $M_3(\mathbb C)$ are decomposable,
\rm Linear Alg. Appl. {\bf 503} (2016), 233--247.


\end{thebibliography}
\end{document}